\documentclass[11pt]{article}

\usepackage[latin1]{inputenc}
\usepackage[T1]{fontenc}
\usepackage{mathptmx}
\usepackage{url}
\usepackage{latexsym}
\usepackage{cite}
\usepackage{multirow}
\usepackage{xcolor,color}
\usepackage{comment}
\usepackage[all]{xy}
\usepackage{mathdots}
\usepackage{setspace}
\usepackage{geometry}
\usepackage{algorithm}
\usepackage{algorithmic}
\usepackage{amsmath, amsthm, amsfonts, amssymb, mathrsfs, amsmath, float}
\usepackage{rotating}
\usepackage{bbm}
\usepackage{xypic}
\usepackage{graphicx}
\usepackage{epstopdf}

\newtheorem{myDef}{Definition}[section]
\newtheorem{defn}[myDef]{Definition}
\newtheorem{prop}[myDef]{Proposition}
\newtheorem{exmp}[myDef]{Example}
\newtheorem{conj}[myDef]{Conjecture}
\newtheorem{lem}[myDef]{Lemma}

\newtheorem{thm}[myDef]{Theorem}
\newtheorem{cor}[myDef]{Corollary}

\floatstyle{ruled}
\newfloat{algorithm}{t}{loa}
\floatname{algorithm}{Algorithm}
\newcommand{\Inp}[1]
  {\noindent\begin{tabular}{@{}p{1.5cm}@{}p{14.0cm}@{}}
   {\bf Input: }&#1 \end{tabular}}
\newcommand{\Outp}[1]
  {\noindent\begin{tabular}{@{}p{1.5cm}@{}p{14.0cm}@{}}
   {\bf Output: }&#1 \end{tabular}}
\newcommand{\SPC}{\hspace*{15pt}}

\def\and{\cap}
\def\bref#1{(\ref{#1})}

\def\tdeg{{{\footnotesize\rm{deg}}}}

\def\Y{{\mathbb{Y}}}

\def\I{\mathcal{I}}
\def\ID{\mathcal{I}}

\def\A{{\mathcal{A}}}

\def\C{{\mathcal{C}}}

\def\P{{\mathbb{P}}}
\def\Q{{\mathbb{Q}}}

\def\gcd{\hbox{\rm{gcd}}}
\def\lcm{\hbox{\rm{lcm}}}

\def\prem{\hbox{\rm prem}}
\def\grem{\hbox{\rm grem}}

\def\sat{\hbox{\rm{sat}}}
\def\asat{\hbox{\rm{asat}}}

\def\deg{\hbox{\rm{deg}}}

\def\init{\hbox{\rm{init}}}
\def\ord{\hbox{\rm{ord}}}

\def\lv{\hbox{\rm{lvar}}}
\def\lead{\hbox{\rm{ld}}}
\def\mod{\hbox{\rm{mod}}}

\def\LC{{\bf LC}}
\def\LM{{\bf LM}}
\def\LT{{\bf LT}}
\def\f{{\bf f}}
\def\g{{\bf g}}
\def\h{{\bf h}}

\def\a{{\bf a}}
\def\b{{\bf b}}
\def\im{{\bf{i}}}

\def\Z{{\mathbb{Z}}}
\def\Q{{\mathbb{Q}}}
\def\N{{\mathbb{N}}}
\def\R{{\mathbb{R}}}

\def\F{{\mathcal{F}}}

\def\lt{\hbox{\rm{lt}}}
\def\lc{\hbox{\rm{lc}}}
\def\lm{\hbox{\rm{lm}}}
\def\grem{\hbox{\rm{grem}}}

\def\Zx{\Z[x]}
\def\Zxn{\Z[x]^n}

\def\fb{{\mathbbm{f}}}

\def\mb{{\mathbbm{m}}}

\def\GB{{\mathbb{G}}}
\def\PH{{\mathbb{P}}}

\begin{document}

\title{Criteria for Finite Difference Gr\"obner Bases of\\
 Normal Binomial Difference Ideals\thanks{Partially supported by a grant from NSFC 11101411.}}
\author{Yu-Ao Chen and Xiao-Shan Gao\\
KLMM, UCAS, Academy of Mathematics and Systems Science\\
The Chinese Academy of Sciences, Beijing 100190, China}
\date{}

\maketitle

\begin{abstract}
\noindent
In this paper, we give decision criteria for normal binomial
difference polynomial ideals in the univariate difference polynomial ring $\F\{y\}$ to have finite difference Gr\"{o}bner bases
and an algorithm to compute the finite difference Gr\"{o}bner bases
if these criteria are satisfied.
The novelty of these criteria lies in the fact that
complicated properties about difference polynomial ideals
are reduced to elementary properties of univariate polynomials in $\Zx$.


\end{abstract}

\section{Introduction}

Difference algebra founded by Ritt and Cohn aims to study algebraic
difference equations in a similar way that polynomial equations are
studied in commutative algebra and algebraic geometry ~\cite{ritt3,cohn,levin,wibmer}.
The Gr\"{o}bner basis invented by Buchberger is a powerful
tool for solving many mathematical problems \cite{buch1}.
%
The concepts of difference Gr\"{o}bner bases was extended to
linear difference polynomial ideals in \cite{levin,klmp,dd-gb}
and nonlinear difference polynomial ideals in \cite{dd-gb}.
Many applications of difference Gr\"{o}bner bases were given  \cite{dd-gb2,app2,levin,klmp}.

Since difference polynomial ideals can be infinitely
generated, their difference Gr\"{o}bner bases are generally infinite.
Even for finitely generated difference polynomial ideals,
their difference Gr\"{o}bner bases could be infinite as shown by Example \ref{ex-11}
in this paper.
This makes it impossible to compute difference Gr\"{o}bner bases
for general difference polynomial ideals
and thus it is a crucial issue to give criteria for difference polynomial ideals to
have finite difference Gr\"{o}bner bases.

Let $\F$ be a difference field and $y$ a difference indeterminate.
In this paper, we will give decision criteria for normal binomial
difference polynomial ideals in $\F\{y\}$ to have finite difference Gr\"{o}bner bases
and an algorithm to compute these finite difference Gr\"{o}bner bases under these criteria.
A difference ideal $\ID$ in $\F\{y\}$ is called normal if
$M P \in \ID$ implies $P \in \ID$ for any difference monomial $M$  in $\F\{y\}$
and $P\in\F\{y\}$.
$\ID$ is called binomial if it is generated by difference polynomials with
at most two terms \cite{dbi,es-bi}.

For $f\in\Zx$,
let $f^{+}, f^{-}\in \N[x]$  be the
positive part and the negative part of $f$ such that $f=f^{+}-f^{-}$.
For  $h=\sum_{i=0}^m a_i x^i \in\N[x]$, denote
$y^h = \prod_{i=0}^m (\sigma^{i} y)^{a_i}$, where $\sigma$ is the difference operator of $\F$. Then any difference monomial in $\F\{y\}$ can be written
as $y^g$ for some $g\in\N[x]$.
For a given $f\in\Zx$ with a positive leading coefficient,
we consider the following binomial difference polynomial ideal in $\F\{y\}$:
$$\ID_f = \sat(y^{f^+} - y^{f^-}) = [\{y^{h^+} - y^{h^-}\,|\, h = gf, g\in\Zx\}]$$
 where
$\sat$ is the difference saturation ideal to be defined in Section 2 of this paper.
Let
\begin{eqnarray*}
\Phi_0 &\triangleq& \{h\in\mathbb{Z}[x]\,|\,\textup{lt}(h)=h^+\}.\cr
\Phi_1 &\triangleq& \{h\in\mathbb{Z}[x]\,|\,hg\in \Phi_0 \textup{ for some monic polynomial } g \in \mathbb{Z}[x]\}.
\end{eqnarray*}
We prove that $\ID_f$ has a finite difference Gr\"obner basis
if and only if $f\in\Phi_1$.
%
This criterion is then extended to general normal binomial difference ideals in $\F\{y\}$.

The decision of $f\in\Phi_1$ is quite nontrivial and
we give the following criteria for $f\in\Phi_1$ based on the roots of $f$:
\begin{enumerate}
\item if $f$ has no positive roots, then $f\in\Phi_1$;
\item if $f$ has more than one positive roots (with multiplicity counted), then $f\not\in\Phi_1$;
\item if $f$ has one positive root $x_+$ and a root $z$
such that $|z| > x_+$, then $f\not\in\Phi_1$;
\item if $f$ has one positive root $x_+$ and a root $z$
such that $|z| = x_+$, then we can compute another $f^*\in\Zx$
and $x^*\in\R_{>0}$
such that $f^*(x^*)=0$, $f^*(w)=0$ and $|w|=x^*$  imply $w=x^*$,
and $f^*(w)=0$ and $|w|\ne x^*$  imply $|w|<x^*$.
Furthermore, $f\in\Phi_1$ if and only if $f^*\in\Phi_1$;

\item if $f\notin\Phi_0$ has a unique positive real root $x_+$ and $x_+<1$,
 then $f\not\in\Phi_1$;
\item
 if $f(1)=0$ and any other root $z$ of $f$ satisfies $|z|<\,$$1$, then $f\in\Phi_1$ if and only if $f(x)/(x-\,$$1)\in\,$$\mathbb{Z}[x^\delta]$ for some $\delta\in\N_{>0}$ and $f(x)(x^\delta-1)/(x-\,$$1)\in\Phi_0$.
\end{enumerate}
With these criteria, only one case is open:
$f$ has a unique positive real root $x_+$, $x_+ > 1$, and $x_+ > |z|$ for any other root $z$ of $f$.
We conjecture that $f\in \Phi_1$ in the above case based on numerical computations.
If $\ID_f$ has a finite difference Gr\"obner basis according to one of the six criteria listed above,
we also give an algorithm to compute it.
%
%

As far as we know the above criteria are the first non-trivial ones
for a difference polynomial ideal to have a finite difference Gr\"obner basis.
The novelty of these criteria lies in the fact that
complicated properties about difference polynomial ideals
are reduced to elementary properties of univariate polynomials in $\Zx$.

The rest of this paper is organized as follows.
In Section 2, preliminaries on Gr\"obner basis for difference polynomial ideals
are given.
In Section 3, criteria for normal binomial difference ideals in $\F\{y\}$ to have finite difference Gr\"{o}bner bases are given.
In Section 4, criteria for $f\in\Phi_1$ and an algorithm to compute
the finite difference Gr\"{o}bner basis of $\ID_f$ under these criteria are given.
In Section 5, we propose an approach based on integer programming
to find $g$ such that $fg\in\Phi_0$ and give a lower bound
for $\deg(g)$ in certain cases.

\section{Preliminaries on Gr\"obner basis of difference polynomial ideals}

\subsection{Gr\"obner basis of a difference polynomial ideal}

An ordinary difference field, or simply a $\sigma$-field, is a field
$\F$ with a third unitary operation $\sigma$ satisfying:  for any
$a, b\in\F$, $\sigma(a+b)=\sigma(a)+\sigma(b)$,
$\sigma(ab)=\sigma(a)\sigma(b)$, and $\sigma(a)=0$ if and only if
$a=0$.
%
%
%
We call $\sigma$ the {\em difference or transforming operator} of $\F$.
A typical example of  $\sigma$-field is
$\Q(\lambda)$ with $\sigma(f(\lambda))=f(\lambda+1)$.
In this paper, we use $\sigma$- as the abbreviation for
difference or transformally.

For $a$ in any $\sigma$-extension ring of $\F$ and $n\in\N_{>0}$,
$\sigma^n(a)$ is called the $n$-th transform of $a$ and denoted by $a^{x^n}$,  with the usual assumption
$a^{0}=1$ and $x^0=1$.
%
%
%
More generally, for $p=\sum_{i=0}^s c_i x^i \in\N[x]$, denote
 $a^p = \prod_{i=0}^s (\sigma^i a)^{c_i}.$
For instance,  $a^{3x^2+x+4}=(\sigma^2(a))^3\sigma(a)a^4$.
It is easy to check that $a^p$ satisfies the properties of powers \cite{dbi}.

Let $S$ be a subset of a $\sigma$-field $\mathcal{G}$ which
contains $\mathcal {F}$. We will denote
$\Theta(S)=\{\sigma^ka|k\in\N, a\in S\}$,   $\mathcal
{F}\{S\}=\mathcal    {F}[\Theta(S)]$.
Now suppose $\Y=\{y_{1},   \ldots,  y_{n}\}$ is a set of
$\sigma$-indeterminates over $\F$.   The elements of $\mathcal
{F}\{\Y\}$ are called {\em $\sigma$-polynomials} over $\F$ in $\Y$.
A {\em $\sigma$-polynomial ideal} $\ID$, or simply a
$\sigma$-ideal, in $\mathcal {F}\{\Y\}$ is a
{\em possibly infinitely generated} ordinary algebraic ideal
satisfying $\sigma(\mathcal {I})\subset\mathcal {I}$.
%
If $S$ is a subset of  $\F\{\Y\}$,  we use $(S)$ and $[S]$  to denote the algebraic ideal and the
$\sigma$-ideal generated by $S$.

%
A monomial order in $\F\{\Y\}$ is called {\em compatible} with the
$\sigma$-structure, if $y_i^{x^{k_1}}<y_j^{x^{k_2}}$ for $k_1 <
k_2$.
Only compatible monomial orders are considered in this paper.
When a monomial order is given, we use $\LM(P)$ and $\LC(P)$ to denote the largest monomial
and its coefficient in $P$ respectively, and $\LT(P)= \LC(P)\LM(P)$ the leading term of $P$.

\begin{defn}\label{def-gb}
$\GB\subset\F\{\Y\}$ is called a {\em $\sigma$-Gr\"obner basis} of
a $\sigma$-ideal $\ID$ if for any $P\in\ID$, there exist $m\in\N$
and  $G\in \GB$ such that $(\LM(G))^{x^m} | \LM(P)$.
\end{defn}

From the definition, $\GB$ is a $\sigma$-Gr\"obner basis of $\ID$ if and only if
$\Theta(\GB)$ is a Gr\"obner basis of $\ID$ treated as an algebraic polynomial ideal
in $\F[\Theta(\Y)]$.
Note that $\ID$ is generally an infinitely generated ideal and the concept of
infinite Gr\"obner basis \cite{Kei} is adopted here.
From this observation, we may see that a $\sigma$-Gr\"obner basis satisfies
most of the properties of the usual algebraic Gr\"obner basis.
For instance, $\GB$ is a $\sigma$-Gr\"obner basis of a $\sigma$-ideal $\ID$ if
and only if for any $P\in\ID$, we have $\grem(P,\Theta(\GB))=0$, where $\grem(P,\Theta(\GB))$ is
the normal form of $P$ modulo $\Theta(\GB)$ in the theory of Gr\"obner basis.
The concepts of reduced $\sigma$-Gr\"obner bases could be similarly introduced.
A $\sigma$-polynomial $Q$ is called {\em $\sigma$-reduced} w.r.t. another $\sigma$-polynomial $P$ if there does not exist a $k\in\N$ such that $\LM(P)^{x^k}$ divides any monomial in $Q$.
Then, a $\sigma$-Gr\"obner $\GB$ basis is called reduced, if any  $P\in \GB$
is $\sigma$-reduced w.r.t $\GB\setminus\{P\}$.
It is easy to see that a $\sigma$-ideal has a unique reduced $\sigma$-Gr\"bner basis.
%

The following example shows that even a finitely generated $\sigma$-ideal may have an infinite $\sigma$-Gr\"obner basis.
As a consequence, there exist no general algorithms to compute the $\sigma$-Gr\"obner basis.

\begin{exmp}\label{ex-11}
Let  $\ID=[y_1y_2^x - y_1^xy_2, y_1y_3-1]$. Assume $y_1< y_2 <y_3$. Then under a compatible monomial order,
the reduced $\sigma$-Gr\"obner basis of $\ID\cap\F\{y_1,y_2\}$ is
%
$\{y_1y_2^{x^{i}}-y_1^{x^i}y_2\,|\, i\in\N_{>0}\}$.
\end{exmp}

\subsection{Characteristic set for a difference polynomial ideal}
\label{sec-cs}
%

The {\em elimination ranking} $\mathscr{R}$ on $\Theta
(\Y)=\{\sigma^ky_i|1\leq i\leq n, k\in\N\}$ is used in this paper:
$\sigma^{k} y_{i}>\sigma^{l} y_{j}$ if and only if $i>j$ or $i=j$
and $k>l$, which is a total order over $\Theta (\Y)$.     By
convention, $1<\sigma^k y_{j}$ for all $k\in\N$.

Let $f$ be a $\sigma$-polynomial in  $\mathcal {F}\{\Y\}$.  The
greatest $y_j^{x^k}$ w.r.t.  $\mathscr{R}$ which  appears
effectively in $f$ is called the {\em leader} of $f$,  denoted by
$\lead(f)$ and correspondingly $y_j$ is called the {\em leading
variable }of $f$, denoted by $\lv(f)=y_j$.
%
%
The leading coefficient of $f$  as a univariate polynomial in
$\lead(f)$ is called the {\em initial} of $f$ and is denoted by
$\init_{f}$.

    Let $p$ and $q$ be two $\sigma$-polynomials in $\F\{\Y\}$.
    $q$ is said to be of higher rank than $p$ if
 $\lead(q)>\lead(p)$  or
 $\lead(q)=\lead(p)=y_j^{x^k}$ and $\deg(q, y_j^{x^k})>\deg(p,  y_j^{x^k})$.
Suppose $\lead(p)=y_j^{x^k}$.  $q$ is said to be {\em Ritt-reduced}
w.r.t. $p$ if  $\deg(q, y_j^{x^{k+l}})<\deg(p, y_j^{x^k})$  for all
$l\in\N$.

 A finite sequence of nonzero $\sigma$-polynomials $\mathcal
{A}:A_1, \ldots, A_m$  is said to be a
    {\em difference ascending chain}, or simply a {\em $\sigma$-chain}, if
 $m=1$ and $A_1\neq0$ or
 $m>1$,  $A_j>A_i$ and $A_j$ is Ritt-reduced  w.r.t. $A_i$ for $1\leq i<j\leq m$.
A $\sigma$-chain $\mathcal{A}$ can be written as the following form \cite{gao-dcs}
\begin{equation}\label{eq-asc}
\mathcal{A}: A_{11}, \ldots, A_{1k_1},\ldots, A_{p1}, \ldots, A_{pk_p}
\end{equation}
where $\lv(A_{ij})=y_{c_i}$ for $j=1, \ldots, k_i$, $\ord(A_{ij},
y_{c_i})<\ord(A_{il}, y_{c_i})$ and
$\deg(A_{ij},\lead(A_{ij})) >  \deg(A_{il},\lead(A_{il}))$
for $j<l$.
%
%
The following are two $\sigma$-chains
\begin{equation}\label{ex-L11}
 \begin{array}{llllll}
 \A_1& :& y_1^{x}-1, &y_1^2y_2^2-1, &y_2^{x}-1& \\
 \A_2& :& y_1^2-1,   &y_1^{x}-y_1, &y_2^2-1, &y_2^{x}+y_2\\
\end{array}
\end{equation}

Let $\mathcal {A}:A_{1},A_{2},\ldots,A_{t}$ be a $\sigma$-chain with
$I_{i}$ as the initial of $A_{i}$, and $P$ any $\sigma$-polynomial.
Then there exists an   algorithm, which reduces
   $P$ w.r.t. $\mathcal {A}$ to a  $\sigma$-polynomial $R$ that is
   Ritt-reduced w.r.t. $\mathcal {A}$ and satisfies the relation
   \begin{equation}\label{eq-prem}
   \prod_{i=1}^t I_{i} ^{e_{i}} \cdot P \equiv
   R, \mod \, [\mathcal {A}],\end{equation}
   where the $e_{i}\in\N[x]$ and
$R=\prem(P,\A)$ is called the {\em
$\sigma$-Ritt-remainder} of $P$ w.r.t. $\A$~\cite{gao-dcs}.

A $\sigma$-chain $\mathcal {C}$ contained in a $\sigma$-polynomial
set $\mathcal {S}$ is said to be a {\em characteristic set} of
$\mathcal {S}$, if  $\mathcal {S}$ does not contain any nonzero
element Ritt-reduced w.r.t. $\mathcal {C}$. Any $\sigma$-polynomial
set has a characteristic set.
A characteristic set
$\mathcal{C}$ of a $\sigma$-ideal $\mathcal {J}$ reduces to zero all
elements of $\mathcal {J}$.

Let $\A:A_1,\ldots,A_t$ be a $\sigma$-chain, $I_i =\init(A_i)$,
$y_{l_i}^{x^{o_i}} = \lead(A_i)$.
$\A$ is called {\em regular} if for any $j\in \N$, $I_i^{x^j}$ is
invertible w.r.t $\A$ \cite{gao-dcs} in the sense that
$[A_1,\ldots,A_{i-1},I_i^{x^j}]$ contains a nonzero
$\sigma$-polynomial involving no $y_{l_i}^{x^{o_i+k}},k=0,1,\ldots$.
To introduce the concept of coherent $\sigma$-chain, we need to
define the {\em $\Delta$-polynomial} first. If $A_i$ and $A_j$ have
distinct leading variables, we define $\Delta(A_i,A_j)=0$. If $A_i$
and $A_j$ ($i<j$) have the same leading variable $y_l$,
$\lead(A_i) = y_l^{x^{o_i}}$, and $\lead(A_j) = y_l^{x^{o_j}}$, then
$o_i < o_j$ \cite{gao-dcs}.
Define
$ \Delta(A_i,A_j) =\prem((A_i)^{x^{o_j-o_i}},A_j).$
 Then $\A$ is called {\em
coherent} if $\prem(\Delta(A_i,A_j),\A)=0$ for all $i< j$
\cite{gao-dcs}.
Both $\A_1$ and $\A_2$ in \bref{ex-L11} are regular and coherent $\sigma$-chains.

Let $\mathcal {A}$ be a $\sigma$-chain. Denote $\mathbb{I}_{\mathcal
{A}}$ to be  the minimal multiplicative set containing the initials
of elements of $\mathcal{A}$ and their transforms.    The {\em
saturation ideal} of $\A$ is defined to be
 $$\sat(\A)=[\mathcal  {A}]:\mathbb{I}_{\mathcal
 {A}} = \{P\in\F\{\Y\}: \exists m\in \mathbb{I}_{\mathcal {A}},  mP\in[A]\}.$$
%
%
 The following result is needed in this paper.
\begin{thm}\cite[Theorem 3.3]{gao-dcs}\label{th-rp}
A $\sigma$-chain $\A$ is a characteristic set of $\sat(A)$ if and
only if $\A$ is regular and coherent.
\end{thm}

We also need the concept of algebraic saturation ideal.
Let $\C$ be an algebraic triangular set in $\F[x_1,\ldots,x_n]$
and $I$ the product of the initials of the polynomials in $\C$.
Then define
$$\asat(\C) = \{P\in\F[x_1,\ldots,x_n]\,|\, \exists k\in\N, I^kP\in(\C)\}.$$

\subsection{$\sigma$-Gr\"obner basis for a binomial $\sigma$-ideal}
\label{sec-bi}

A $\sigma$-monomial in $\Y$ can be  written as $\Y^\f =
\prod_{i=1}^n y_i^{f_i}$, where $\f=(f_1, \ldots, f_n)^\tau\in \N[x]^{n}$.
A nonzero vector $\f=(f_1, \ldots, f_n)^\tau\in\Z[x]^n$ is said to be {\em
normal} if  the leading coefficient of $f_s$ is positive,  where $s$
is the largest subscript such that $f_s\ne0$.
For $\f\in\Z[x]^n$,  let $\f^{+}, \f^{-}\in \N^{n}[x]$ denote respectively the
positive part and the negative part of $\f$ such that $\f=\f^{+}-\f^{-}$.
Then  $\gcd(\Y^{\f+},\Y^{\f^-})=1$ for any $\f\in\Zxn$.
If $\f\in\Z[x]^n$ is normal, then $\Y^{\f^+} > \Y^{\f^-}$ and $\LT(\Y^{\f+}-c\Y^{\f^-})= \Y^{\f+}$ under a monomial
order compatible with the $\sigma$-structure.

A {\em $\sigma$-binomial}   in  $\Y$ is  a $\sigma$-polynomial with
at most two terms,  that is, $a\Y^{\a}+b\Y^{\b}$ where $a, b\in \F$
and $\a, \b \in\N[x]^n$.
A $\sigma$-ideal in $\F\{\Y\}$ is called {\em binomial} if it is
generated by, possibly infinitely many, $\sigma$-binomials \cite{dbi}.
We have
\begin{prop}[\cite{dbi}]\label{prop-1}
A $\sigma$-ideal $\ID$ is binomial if and only if the reduced
$\sigma$-Gr\"{o}bner basis for $\ID$ consists of $\sigma$-binomials.
\end{prop}

%

Let $\mb$ be the multiplicative set generated by $y_i^{x^j}$ for
$i=1, \ldots, n,  j\in\N$.
A $\sigma$-ideal $\ID$ is called {\em normal} if for $M\in\mb$
and $P\in\F\{\Y\}$,  $MP\in\I$ implies $P\in\I$.
Normal $\sigma$-ideals in $\F\{\Y\}$ are closely related with the
$\Z[x]$-modules in $\Zxn$ \cite{GHNF-alg,dbi}, which will be explained below.
We first introduce a new concept.

\begin{defn}
A {\em partial character} $\rho$ on $\Z[x]^{n}$ is a homomorphism
from a $\Z[x]$-module $L_{\rho}$ in $\Zxn$ to the multiplicative group
$\F^{\ast}$ satisfying $\rho(x\f)=(\rho(\f))^x=\sigma(\rho(\f))$ for $\f\in
L_\rho$.
\end{defn}

A $\Zx$-module generated by $\h_1,\ldots,\h_m\in \Zxn$ is denoted as
$(\h_1,\ldots,\h_m)_{\Zx}$.
Let $\rho$ be a partial character over $\Zxn$ and $\fb=\{\f_1,\ldots,\f_s\}$ a reduced
Gr\"obner basis of the $\Zx$-module $L_{\rho}=(\fb)_{\Zx}$.
For $\h\in \Zxn$ and $H\subset L_\rho$, denote $\PH_\h=\Y^{\h^+} - \rho(h) \Y^{\h^-}$
and $\PH_H = \{\PH_\h \,|\, \h\in H\}$.
Introduce the following notations associated with $\rho$:
 \begin{eqnarray}\label{eq-A+}
&& \I^{+}(\rho):= [\PH_{L_\rho}]= [{\Y^{\f^{+}}-\rho(\f)\Y^{\f^{-}}\,|\,\f\in L_{\rho}}]\\
&& \A^{+}(\rho):= \PH_{\fb} = \{\Y^{\f_1^+}-\rho(\f_1)\Y^{\f_1^-},
      \ldots,
      \Y^{\f_s^+}-\rho(\f_s)\Y^{\f_s^-}\}.
 \end{eqnarray}
It is shown that \cite{dbi} $\A^{+}(\rho)$ is a regular and coherent $\sigma$-chain and hence is a characteristic set of $\sat(\A^{+}(\rho))$ by Theorem \ref{th-rp}. Furthermore, we have
\begin{thm}\label{th-nbgb1} The following conditions are equivalent.
\begin{enumerate}
\item $\ID$ is a normal binomial $\sigma$-ideal in $\F\{\Y\}$.
\item $\ID=\I^{+}(\rho)$ for a partial character $\rho$ over $\Zxn$.
\item $\ID=\sat(\A^{+}(\rho))$  for a partial character $\rho$ over $\Zxn$.
%
\end{enumerate}
Furthermore, for $\f\in \Zxn$, $\Y^{\f^{+}}-c\Y^{\f^{-}}\in\ID \Leftrightarrow \f\in L_{\rho}$ and $c=\rho(\f)$.
\end{thm}

As a direct consequence of Proposition \ref{prop-1} and Theorem \ref{th-nbgb1}, we have
\begin{cor}\label{cor-nbgb1}
Let $\rho$ be a partial character over $\Zxn$.
Then $\PH_{L_\rho}$ is a $\sigma$-Gr\"obner basis of $\ID^+(\rho)$.
\end{cor}
Note that for $\f\in \Zxn$, either $\f$ or $-\f$ is normal and
we need only consider the normal vectors in the $\sigma$-Gr\"obner basis.
%
So, for simplicity, we may assume that all given vectors are normal.
We have the following criterion for the $\sigma$-Gr\"obner basis of normal binomial $\sigma$-ideals.
\begin{cor}\label{cor-nbgb2}
Let $\rho$ be a partial character over $\Zxn$ and $H\subset L_\rho$.
Then $\PH_H$ is a $\sigma$-Gr\"obner basis of $\ID^+(\rho)$ if and only if
for any normal $\g \in L_\rho$, there exist $\h\in H$ and $j\in \N$, such that
$\g^+ - x^j \h^+ \in\N[x]^n$.
\end{cor}
\begin{proof}
By Corollary \ref{cor-nbgb1}, $\PH_{L_\rho}$ is a $\sigma$-Gr\"obner basis of $\ID^+(\rho)$.
Then  $\PH_{H}$ is a $\sigma$-Gr\"obner basis of $\ID^+(\rho)$ if and only if
for any normal $\g \in L_\rho$, there exist $\h\in H$ and $j
\in \N$ such that
$\LM(x^j\PH_\h)|\LM(\PH_\g)$, which is equivalent to
$\g^+ - x^j \h^+ \in\N[x]^n$.
\end{proof}

\begin{exmp}\label{ex-2}
Let $\f=[1-x,x-1]$, $L=(\f)_{\Zx}$, and $\rho$ the trivial
partial character on $L$, that is, $\rho(\h)=1$ for $\h\in$L.
Then $\P_\f=y_1y_2^{x}-y_1^{x}y_2$.
By Theorem \ref{th-nbgb1}, $\I^{+}(\rho)=\sat(\P_\f)$.
By Corollary \ref{cor-nbgb1}, a $\sigma$-Gr\"obner basis of $\ID^+(\rho)$
is $\{ \Y^{\g^+}- \Y^{\g^-} \,|\, \g=h\f, h \in\Z[x], \lc(h)>0 \}$.
By Example \ref{ex-11},  $\sat(\P_\f)= [\P_\f,y_1y_3-1]\cap\Q\{y_1, y_2\}=
[y_1y_2^{x^{i}}-y_1^{x^{i}}y_2\,|\,  i\in\N_{>0}]$, and
a reduced $\sigma$-Gr\"obner basis of $\ID^+(\rho)$ is
$\{y_1y_2^{x^{i}}-y_1^{x^{i}}y_2\,|\,  i\in\N_{>0}\}$.
\end{exmp}

\section{Criteria for finite $\sigma$-Gr\"{o}bner basis}
In this section, we will give a criterion for the $\sigma$-Gr\"obner basis of
a normal binomial $\sigma$-ideal in $\F\{y\}$ to be finite,
where $y$ is a $\sigma$-indeterminate.
Without loss of generality, we assume $\rho(h) =1$ for all partial characters $\rho$ over $\Zx$ and $h\in L_\rho$.

\subsection{Case 1: characteristic set contains a single $\sigma$-polynomial}
In this section, we consider the simplest case: $n=1$ and $L_\rho=(f)_{\Zx}$ is generated by one polynomial $f\in\Zx$.
We will see that even this case is highly nontrivial.
For $g\in\Zx$, we use $\lc(g)$, $\lm(g)$, and $\lt(g)$ to represent the
leading coefficient, leading monomial, and leading term of $g$, respectively.

In the rest of this section, we assume $f\in\Zx$ and $\lc(f) >0$.
Then $\PH_f = y^{f^+}-y^{f^-}$ and $\LT(\PH_f) = y^{f^+}$ under a monomial order compatible with the $\sigma$-structure.
By Theorem \ref{th-nbgb1}, all normal binomial $\sigma$-ideals in $\F\{y\}$ whose
characteristic set consists of a single $\sigma$-polynomial  can be written as the following form:
\begin{eqnarray}\label{eq-idf1}
\ID_f = \sat(\PH_f)=[y^{h^+}-y^{h^-}\,|\, h = fg \in(f)_{\Zx}, \forall (g\in \Z[x],\lc(g)>0) ].
\end{eqnarray}
In this section, we will give a criterion for $\ID_f$
to have a finite $\sigma$-Gr\"obner basis. Define
\begin{eqnarray}
\Phi_0 &\triangleq& \{f\in\mathbb{Z}[x]\,|\,\textup{lt}(f)=f^+\}.\cr
\Phi_1 &\triangleq& \{f\in\mathbb{Z}[x]\,|\,fg\in \Phi_0 \textup{ for some monic polynomial } g \in \mathbb{Z}[x]\}.
\end{eqnarray}
We now give the main result of this section, which can be deduced from Lemma \ref{lm-s2} and Lemma \ref{lm-s4}.
\begin{thm}\label{th-sg1}
 $\ID_f$ in \bref{eq-idf1} has a finite $\sigma$-Gr\"obner basis under a monomial order compatible
 w.r.t the $\sigma$-structure if and only if $f\in\,$$\Phi_1$.
\end{thm}

For two polynomials $h_1$ and $h_2\in\,$$\mathbb{Z}[x]$, denote  $h_1\succeq\,$$h_2$ if $h_1-\,$$h_2\in\mathbb{N}[x]$.
For $h_1$ and $h_2\in\,$$\mathbb{N}[x]$, we have $h_1\succeq\,$$h_2$ if and only if $y^{h_2}\,|\,y^{h_1}$.
\begin{lem}\label{lm-s1}
If $f\in \Phi_0$, then $\{\PH_f\}$ is a $\sigma$-Gr\"obner basis  of $\ID_f$.
\end{lem}
\begin{proof}
For $g\in(f)_{\Zx}$ with $\textup{lc}(g)>\,$$0$,  $\exists\,h\in\,$$\mathbb{Z}[x]$ with $\textup{lc}(h)>\,$$0$ such that $g=\,$$fh$. Since $f\in \Phi_0$, we have $\lt(f) = f^+$. Then,
$$x^{\footnotesize\deg(h)}f^+= \textup{lt}(h)f^+/\textup{lc}(h)\preceq \textup{lt}(h)f^+=\textup{lt}(h)\textup{lt}(f)= \textup{lt}(g)\preceq g^+.$$
By Corollary \ref{cor-nbgb2}, $\{\PH_f\}$ is a $\sigma$-Gr\"obner basis  of $\ID_f$.
\end{proof}

\begin{lem}\label{lm-s2}
If $f\in\,$$\Phi_1$, then $\ID_f$ has a finite $\sigma$-Gr\"obner basis.
\end{lem}
\begin{proof}
Let $h=fg\in\,\Phi_0$, where $g$ is monic.
Then $\lc(h) = \lc(f)$ and $\lt(h) = \lt(f)\lm(g)= h^+$.
$\ID_{\footnotesize{\textup{deg}}(h)}=\ID_f\bigcap \F[y,y^x,\cdots,y^{x^{\tiny{\textup{deg}}(h)}}]$
is a polynomial ideal in a polynomial ring with finitely many variables, which has a finite
Gr\"obner basis denoted by $\mathbb G_{\leqslant \footnotesize{\textup{deg}}(h)}$.
Let $\PH_u\in\ID_f$ and $\textup{lc}(u)>\,$$0$. If $\textup{deg}(u)\leqslant\,$$\textup{deg}(h)$,
then there exists a $\PH_t\in G_{\leqslant \footnotesize{\textup{deg}}(h)}$ such that  $t\preceq u$.
Otherwise, we have $\deg(u) > \deg(h)$ and $\lc(u) \ge \lc(f)$. Then
\begin{eqnarray*}
x^{\footnotesize{\textup{deg}}(u)-\footnotesize{\textup{deg}}(h)}h^+
&=&x^{\footnotesize{\textup{deg}}(u)-\footnotesize{\textup{deg}}(f)-\footnotesize{\textup{deg}}(g)}\textup{lt}(f)\textup{lm}(g)\\
&=&x^{\footnotesize{\textup{deg}}(u)-\footnotesize{\textup{deg}}(f)}\textup{lt}(f) = x^{\footnotesize{\textup{deg}}(u) - \footnotesize{\textup{deg}}(f)} \lc(f) \lm(f)
=  \lc(f) \lm(u)
\preceq \textup{lt}(u)\preceq u^+.
\end{eqnarray*}
Since that $\PH_h\in\ID_{\footnotesize{\textup{deg}}(h)}$,
by Corollary \ref{cor-nbgb2},  $\mathbb G_{\leqslant \footnotesize{\textup{deg}}(s)}$ is a finite $\sigma$-Gr\"obner basis  of $\ID_f$.
\end{proof}
\begin{cor}\label{cor-s2}
Let $f\in\,$$\Phi_1$, $h=gf\in\Phi_0$, $g$ a monic polynomial in $\Zx$ , and $D=\deg(h)$.
Then the Gr\"obner basis of the polynomial ideal
$\ID_{D}=\ID_f\bigcap \F[y,y^x,\cdots,y^{x^D}]$
is a  finite $\sigma$-Gr\"obner basis for $\ID_f$.
\end{cor}

From the proof of Lemma \ref{lm-s2}, we have
\begin{exmp}
$f = x^2 + x + 1\in\Phi_1$, because  $(x-1)f = x^3-1\in\Phi_0$.
The finite $\sigma$-Gr\"obner basis is $\GB = \{y^{x^2 + x + 1} -1, y^{x^3}-y\}$.
\end{exmp}

Let $D$ be $\R$ or $\Z$. We will use the following new notation
$$D^{>0}[x]\triangleq\{\sum _{i=0}^na_ix^i\,|\,n\in\N,\,\forall i(a_i\in D_{>0})\}.$$

\begin{lem}\label{lm-s3}
$\mathbb{N}[x]\subseteq \Phi_1$.
\end{lem}
\begin{proof}
Let $g=a_nx^n+a_{n-1}x^{n-1}+\cdots+a_0\in\,$$\mathbb{N}[x]$ with $d$$=\textup{max}\{d\in\,$$\mathbb{N}\,|\,x^d\,|\,g\}$ the multiplicity of $f$ at $0$. Then $a_d>\,$$0$. Let $s=\,$$(x^{n-d}+x^{n-d-1}+\cdots+1)g=a_nx^{2n-d}+(a_n+a_{n-1})x^{2n-d-1}+\cdots+(a_n+\cdots+a_d)x^{n}+(a_{n-1}+\cdots+a_d)x^{n-1}+\cdots+a_dx^d$. %
Rewrite $s = b_{2n-d}x^{2n-d}+\cdots+b_dx^d$.
Then $s/x^d\in\Z^{>0}[x]$.
Let $M= \lceil\textup{max}\{b_{i-1}/b_i\,|\,d+1\leqslant i\leqslant 2n-d\} \rceil +1$.
Then $(x-M)s=b_{2n-d}x^{2n-d+1}+(b_{2n-d-1}-Mb_{2n-d})x^{2n-d}+\cdots+(b_d-Mb_{d+1})x^{d+1}-Mb_dx^d\in\Phi_0$.
So both $s$ and $g$ are in $\Phi_1$.
\end{proof}

\begin{lem}\label{lm-s4}
If $f \not\in \Phi_1$, then $\ID_f$ does not have a finite $\sigma$-Gr\"obner basis.
\end{lem}
\begin{proof}
Suppose otherwise, $\ID_f$ has a finite $\sigma$-Gr\"obner basis  $\mathbb G = \PH_H$,
where $H=\{f_1,\cdots,f_l\}\subset\Zx$ with each $\textup{lc}(f_i)>0$.
Since $f$ has the lowest degree in $(f)_{\Zx}$, we have $f\in H$.

Let $H_c\triangleq \{h\in H\,|\,\lc(h) = \lc(f)\}$.
Since $f\notin \Phi_1$, we have $H_c\bigcap \Phi_1 = \emptyset$.
By Lemmas \ref{lm-s1} and \ref{lm-s3}, for all $h \in H_c$,
$h^+$ has at least two terms and  $h^-$ has at least one term.
For $u\in\Zx$ with $\lc(u)>0$,
define a function
\begin{equation}\label{eq-deg}
\widetilde{\deg}(u) = \deg(u) - (\deg(u^+ - \lt(u)))
\end{equation}
which is the degree gap between the first two highest monomials of $u^+$.
Suppose $h_1$ is an element in $H_c$ such that
$\widetilde{\deg}(h_1) = \max \{\widetilde{\deg}(h)\,|\, h\in H_c\}$.
$h_1$ exists because $f\in H_c\neq \emptyset$ and $H_c$ is a finite set.
Denote $\textup{lt}(h_1)\triangleq ax^n$, $\tilde h_1\triangleq h_1-\textup{lt}(h_1)$, $\textup{lt}(\tilde h_1^+)\triangleq b x^m$, and $\tilde{\tilde h}_1^+\triangleq \tilde h_1^+-\textup{lt}(\tilde h_1^+)$. Then $h_1=ax^n+bx^m+\tilde{\tilde h}_1^+ - h_1^-$.
Since $h_1\not\in\Phi_1$, we have $ab >0$.
Let $c\triangleq \lceil b/a\rceil \ge 1$ and
$$s=
(x^n-cx^m)h_1=ax^{2n}+x^n\tilde{\tilde h}_1^++cx^mh_1^--(ac-b)x^{m+n}-cx^m\tilde h_1^+-x^nh_1^-.$$
We have $s^+\preceq s_0\triangleq ax^{2n}+x^n\tilde{\tilde h}_1^++cx^mh_1^-$,
and  $\widetilde{\deg}(s) = \deg (s)-\deg (s^+-\textup{lt}(s))\ge
\widetilde{\deg}(s_0)$
$ =\deg (s_0)-\deg (s_0^+-$ $\textup{lt}(s_0))>n-m$
$= \widetilde{\deg}(h_1) = \deg (h_1)-\deg (h_1^+-\textup{lt}(h_1)).$

Since $\PH_H$ is a $\sigma$-Gr\"obner basis of $\ID_f$,
there exist $h\in H$ and $j\in\N$
such that $t=s^+ - x^jh^+\in\N[x]$.
We claim $\lt(t) = \lt(s^+)$.
If $h\in H_c$, then $\widetilde{\deg}(s)> \widetilde{\deg}(h)$.
Note that $\deg(s^+)= \deg(x^jh)$  implies that
the coefficient of the second largest monomial of $s^+ - x^jh$ is negative
contradicting to the fact  $s^+ - x^jh\in\N[x]$.
As a consequence, we must have $\deg(s^+)>  \deg(x^jh)$ and the claim is proved in this case.
Now let $h\in H\backslash H_c$.
Since $\textup{lc}(h)>\textup{lc}(s)=\textup{lc}(f)$,
we have $\deg (x^jh)<\deg (s)$ which
implies $\lt(t) = \lt(s^+)$. The claim is proved.
The fact $\lt(t) = \lt(s^+)$ implies that when computing the
normal form $\PH_u=\grem(\PH_s,\Theta(\PH_{H}))$, we always have
$\lt(u) = \lt(s)$. As a consequence, $\PH_u\ne0$ which contradicts to
the fact that $\PH_{H}$ is a $\sigma$-Gr\"obner basis of
$\ID_f$ and $s\in(f)_{\Zx}$.
\end{proof}

Note that the proof of Lemma \ref{lm-s4} gives a method to
construct infinitely many elements in a $\sigma$-Gr\"bner basis as
shown in the following example.
\begin{exmp}\label{ex-s2}
Let $f = x^2-2 x + 1\notin\Phi_1$.
In the proof of Lemma \ref{lm-s4}, $c= \lceil b/a\rceil = 1$ and
$s_1 = (x^{2}-1)f = x^4 + 2x - 2x^3 -1$.
Repeat the above procedure to $s_1$, we obtain
$s_2 = (x^4-2x)s_1 = x^8 + 3x^4+2x - 2x^7-4x^2$.
Then $\widetilde{\deg}(f) < \widetilde{\deg}(s_1) < \widetilde{\deg}(s_2)$
and  $\PH_{s_i}$ is in a $\sigma$-Gr\"obner basis for all $i$.
Thus any $\sigma$-Gr\"obner basis of $\ID_f$ is infinite.
We can show that a minimal $\sigma$-Gr\"{o}bner basis is $\GB = \{y^{x^{2i}+1}-y^{2x^i}\,|\,i \in\Z_{>0}\}\bigcup\{y^{x^{2i+1}+1} -y^{x^{i+1}+x^i}\,|\,i\in\Z_{>0}\}$.
\end{exmp}

\subsection{Finite $\sigma$-Gr\"{o}bner bases for normal binomial $\sigma$-ideals}
In this section, we consider the general normal binomial $\sigma$-ideals in $\F\{y\}$.
By Theorem \ref{th-nbgb1}, all normal binomial $\sigma$-ideals in $\F\{y\}$
can be written as the following form:
\begin{eqnarray}\label{eq-idf}
\ID_\GB = \sat(\PH_\GB)=[y^{g^+}-y^{g^-}\,|\, \forall g  \in (\GB)_{\Zx}, \lc(g)>0]
\end{eqnarray}
where
\begin{eqnarray}\label{eq-gb2}
\GB = \{g_1,\ldots,g_t\}\subset \Zx
\end{eqnarray}
is a reduced Gr\"obner basis of the $\Zx$-module $L=(\GB)_{\Zx}$.
Gr\"obner bases in $\Zx$ have the following special structure \cite{dbi}.
\begin{lem}\label{lm-21}
Let  $\GB = \{g_1,\ldots,g_k\}$ be a reduced Gr\"{o}bner basis of a
$\Zx$-module in $\Z[x]$, $g_1<\cdots<g_k$, and
$\lt(g_i)=c_ix^{d_i}\in\N[x]$. Then

1) $0\le d_1 < d_2 < \cdots<d_k$.

2) $c_k | \cdots | c_2 | c_1 $ and $c_i \ne c_{i+1}$ for $1\le i\le
k-1$.

3) $\frac{c_i}{c_k} | g_i$ for $1\le i< k$. If
$\widetilde{b}_1$ is the primitive part of $g_1$, then
$\widetilde{b}_1 | g_i$ for $1< i\le k$.
%

\end{lem}
Here are two Gr\"{o}bner bases in $\Z[x]$\hbox{\rm:} $\{ 4, 2x \}$,  $ \{15, 5x, x^2+3\}$.

In the rest of this section, let $L=(\GB)_{\Zx}$ for
$\GB$ defined in \bref{eq-gb2} and define
%
\begin{eqnarray}
L_i&\triangleq&\{f\in L\,|\,\lc(f)=c_t=\lc(g_t)\}\\
L_t&\triangleq&\{f\in L_i\,|\,f \hbox{ has minimal degree in } L_i\}\}.
\end{eqnarray}

\begin{thm}\label{lm-mg1}
$\mathcal{I}_{\GB}$ has a finite $\sigma$-Gr\"{o}bner basis  if and only if $L_i\bigcap\Phi_0\ne\emptyset$.
\end{thm}
\begin{proof}
Suppose $L_i\bigcap\Phi_0\ne\emptyset$ and let $g\in L_i\bigcap\Phi_0$. Then $\ID_\GB\bigcap k[y,y^x,\cdots,y^{x^{\tiny{\deg }(g)}}]$ has a finite
Gr\"obner basis denoted by $G_{\le\tdeg(g)}$.
Let $\PH_u\in\ID_\GB$ and $\textup{lc}(u)>0$. If $\deg (u)\le\deg (g)$, then there exists a $\PH_h\in G_{\le\tdeg(g)}$ such that $h\preceq u$.
Otherwise, we have $\deg(u) > \deg(g)$ and $\lc(u) \ge \lc(g)$. Then
\begin{eqnarray*}
x^{\tdeg(u)-\tdeg(g)}g^+
=x^{\tdeg(u)-\tdeg(g)}\textup{lt}(g) = x^{\tdeg(u)-\tdeg(g)} \lc(g) \lm(g)
=  \lc(g) \lm(u)
\preceq \textup{lt}(u)\preceq u^+.
\end{eqnarray*}
By Corollary \ref{cor-nbgb2},  $\mathbb G_{\le  \tdeg(g)}$ is a finite $\sigma$-Gr\"obner basis  of $\ID_\GB$, since $\PH_g$ is in $\mathbb G_{\le  \tdeg(g)}$.

We will prove the other direction by contradiction.
Suppose that $L_i\cap\Phi_0=\emptyset$ and
$\ID_\GB$ has a finite $\sigma$-Gr\"{o}bner basis $\PH_H=\{\PH_{u_1},\cdots,\PH_{u_k}\}$.
Let $H=\{{u_1,\cdots,u_k}\}$, and $H_c=H\bigcap L_i$.
Since $\grem(\PH_{g_t},\Theta(\PH_H))=0$, we have
$H_c\ne\emptyset$ and let $u_1$ be an element of $H_c$ with maximal $\widetilde{\deg}$
which is defined in \bref{eq-deg}.
Since $L_i\cap\Phi_0=\emptyset$, by Lemma \ref{lm-s3} $u_1^+$ contains at least two terms
and $u_1^-\ne0$.
Similar to the proof of Lemma \ref{lm-s4}, we can construct
an $s\in\Zx\cap L$ such that $\widetilde{\deg}(s) > \widetilde{\deg}(u_1)$
and $\lc(s)=\lc(u_1)$. Then, $\grem(\PH_s,\Theta(\PH_H))\ne0$ contradicting
to the fact that $\PH_H$ is a $\sigma$-Gr\"obner basis.
%
%
\end{proof}

\begin{cor}
If $\mathcal{I}_\GB$ has a finite $\sigma$-Gr\"{o}bner basis,
then $g_1\in\Phi_1$.
\end{cor}
\begin{proof}
Let $\widetilde{b}_1$ be the primitive part of $g_1$.
Then by Lemma \ref{lm-21}, $\widetilde{b}_1|h$ for any $h\in L$.
By Theorem \ref{lm-mg1}, $\widetilde{b}_1$ and hence $g_1$ is in $\Phi_1$.
\end{proof}

\begin{cor}
If $L_t\bigcap\Phi_1\ne\emptyset$ and in particular $g_t\in\Phi_1$, then $\mathcal{I}_\GB$ has finite $\sigma$-Gr\"{o}bner Basis.
\end{cor}

The following example shows that $g_t\in\Phi_1$ is not a necessary condition for the $\sigma$-Gr\"obner basis to be finite.
\begin{exmp}
Let $\GB=\{2(x^2-2),(x^2-2)(x+1)\}$.
Then  $(x^2-2)(x+1)(x-1)+2(x^2-2)=x^4-x^2-2\in\Phi_0\subset\Phi_1$,
and hence $\ID_{\GB}$ has a finite $\sigma$-Gr\"obner basis. On the other hand, we
will show  $(x^2-2)(x+1)\notin\Phi_1$ in Example \ref{ex-51}.
\end{exmp}

In order to give another criterion,  we need the following effective Polya Theorem.
\begin{lem}[\cite{poliya}]\label{lm-Po}
Suppose that $f(x)=\sum\limits_{j=0}^na_nx^n\in\R[x]$ is positive on $[0,\infty)$
and $F(x,y)$ the homogenization of $f$.
Then for $N_f> \frac{n(n-1)L}{2\lambda} - n$, $(1+x)^{N_f}f(x)\in\R^{>0}[x]$,
where $\lambda = \min\{F(x,1-x)\,|\, x\in[0,1]\}$
and $L = \max\{\frac{k!(n-k)!}{n!}|a_k|\}$.
\end{lem}

\begin{cor}
If there exists an $h\in L$ with no positive real roots, then $\mathcal{I}_\GB$ has a finite $\sigma$-Gr\"{o}bner basis.
\end{cor}
\begin{proof}
Write $h = x^{m_1}h_1 $ such that $h_1(0)\ne0$.
By Lemma \ref{lm-Po}, there exists an $N\in\N$ such that $h_2=(x+1)^N h\in\Z^{>0}[x]$.
Take a sufficiently large $N$ such that $\deg(h_2) > d_t=\deg(g_t)$.
%
Then there exists a sufficiently large $M\in\mathbb{N}$, such that
$\overline{g}=x^{m_1}(x^{\deg(h_2)-\deg(g_t)+1}g_t- Mh_2)\in \Phi_0$.
Since $\overline{g}\in L_i$, by Lemma \ref{lm-mg1}, $\mathcal{I}$ has a finite $\sigma$-Gr\"{o}bner Basis.
\end{proof}


\section{Membership decision for $\Phi_1$ and $\sigma$-Gr\"obner basis computation}
In Section 3, we prove that $\textup{sat}(\PH_f)$ has a finite $\sigma$-Gr\"obner basis
if and only if $f \in \Phi_1$. In this section, we will give criteria and an algorithm for $f \in \Phi_1$.
If $f \in \Phi_1$, we also give an algorithm to compute the finite $\sigma$-Gr\"obner basis.

From the definition of $\Phi_1$, a necessarily condition for $f\in \Phi_1$ is $\text{lc}(f)>0$.
Also, it is easy to show that $f\in \Phi_1$  if and only if  $cx^mf\in \Phi_1$ for positive integers
$c$ and $m$.
So in the rest of this paper, we assume
$$f=\sum_{i=0}^n a_n x^i \in\Zx$$
such that $n>0$, $\lc(f)=a_n>0$, $f(0)=a_0\ne0$, and $\gcd(a_0,a_1,\ldots,a_n)=1$.

\subsection{Decision criteria}
In this subsection, we will study whether $f\in \Phi_1$ by examining properties of the roots of $f(x)=0$.

\begin{lem}\label{cor-r1}
If $f\in \mathbb{Z}[x]$ has no positive real roots, then $f\in \Phi_1$.
\end{lem}
\begin{proof}
By Lemma \ref{lm-Po}, there exists an $N\in\,$$\N$, such that $(x+1)^Nf\in\,$$\Z^{>0}[x]\subseteq\N[x]$. By Lemma \ref{lm-s3}, $(x+1)^Nf\in\,$$\N[x]\subseteq \Phi_1$, and thus $f\in \Phi_1$.
\end{proof}

By Lemma \ref{cor-r1}, we need only consider those polynomials  which have positive roots.
\begin{lem}\label{lm-8}
Let $f=a_nx^n+\cdots+a_0 \in\Phi_0$. Then
$f$ has a simple and unique positive real root $x_+$,
and for any root $z$ of $f$, we have $|z|\le x_+$.
\end{lem}
\begin{proof}
Since $f\in\Phi_0\setminus\Z$, the number of sign differences of $f$ is one.
Then by Descartes' rule of signs\cite{drule}, the number of positive real roots of $f$
(with multiplicities counted) is one or less than one by an even number.
Then $f$ has a simple and unique positive real root $x_+$.
For any root $z$ of $f$, since $-a_i\ge 0$ for $i=0,\ldots,n-1$, we have
\begin{equation}\label{eq-cr1}
a_n|z|^{n}=|a_n z^{n}|=|-a_{n-1}z^{n-1}-\cdots-a_0|\le -a_{n-1}|z|^{n-1}-\cdots-a_0.
\end{equation}
Thus $f(|z|)\leqslant\,$$0$ and hence $f$ has at least one real root in $[|z|,\infty)$.
Since $f$ has a unique positive real root $x_+$, we have $|z|\leqslant\,$$x_+$.
\end{proof}

We now consider those $f$ which has a root $z\ne x_+$ and $|z|=x_+$.
Such a $z$ must be either $-x_+$ or a complex root.
\begin{lem}\label{lm-81}
Let $f=a_nx^n+\cdots+a_0 \in\Phi_0$ and $x_+$ the unique positive root of $f$.
If $f$ has a root $z\ne x_+$ but $|z|=x_+$, then we have
\begin{enumerate}
\item
$z^{\delta_f}\in\R_{>0}$  and $z$ is a simple root of $f$,
where  $\delta_f= \gcd\{i\,|\, a_i\ne0\}>1$.
\item
$f$ is a polynomial in $x^{\delta_f}$: $f=\widehat{f}\circ x^{\delta_f}$, where $\circ$ is the function composition.
Furthermore, $\widehat{f}(w)=0$ and $|w|=x_+^{\delta_f}$ imply $w=x_+^{\delta_f}$.
\item
$f$ has exactly $\delta_f$ roots with absolute value $x_+$: $\{z\,|\, f(z)=0, |z|=x_+\} = \{\zeta^k x_+ \,|\, \zeta
 = e^{\frac{2\pi\im}{\delta_f} },k=1,\ldots,\delta_f \}$, where $\im=\sqrt{-1}$.
\end{enumerate}
\end{lem}
\begin{proof}
Let $z\ne x_+$ be a root of $f$ such that $|z|=x_+$.
Then $f(|z|)=f(x_+)=a_n|z|^n+a_{n-1}|z|^{n-1}+\cdots+a_0=0$,
which, combining with \bref{eq-cr1}, implies
$|-a_{n-1}z^{n-1}-\cdots-a_0|= -a_{n-1}|z|^{n-1}-\cdots-a_0$.
The above equation is possible if and only if
$-a_iz^i\in\R_{>0}$ for each $i\le n-1$ and $a_i\ne0$.
Also note, $z^n= (-a_{n-1}|z|^{n-1}-\cdots-a_0)/a_n\in \R_{>0}$.
Then,  $z^i\in\R_{>0}$ for each $i\le n$ and $a_i\ne0$.
%
%
Note that $z^{m}\in\R_{>0}$ and  $z^{k}\in\R_{>0}$ imply $z^{m-k}\in\R_{>0}$.
As a consequence, $z^{\delta_f} \in\R_{>0}$ for $\delta_f=\gcd\{ i\,|\,a_i\ne0\}$.
Since $z\ne x_+$, we have $\delta_f> 1$.
Part 1 of the lemma is proved.

From the definition of $\delta_f$, $f$ is a polynomial of $x^{\delta_f}$:
$f(x) = \widehat{f}(x)\circ (x^{\delta_f})$.
It is easy to see that $\widehat{f}(x)\in\Phi_0$.
Let $\widehat{f}(x)=b_k x^k + \cdots+ b_1 x + b_0$.
Then $\gcd\{j\,|\, b_j\ne0\}=1$.
By the first part of this lemma, we know $x_+^{\delta_f}$ is the only root of $f$ whose absolute value
is $x_+^{\delta_f}$.
Since $z^{\delta_f}$ and $x_+^{\delta_f}$ are both the unique positive real roots of $\widehat{f}(x)$,
we have $z^{\delta_f}=\,$$x_+^{\delta_f}$ and hence $z$ is a simple root of $f$.
Part 2 of the lemma is proved.
Part 3 of the lemma comes from the fact $z^{\delta_f}=\,$$x_+^{\delta_f}$ is
the unique positive real root of $f$
and $f(z)=\widehat{f}(z^{\delta_f})=0$.
\end{proof}

\begin{cor}\label{cor-af8}
If $f\in\,$$\Phi_1$ has at least one positive real root $x_+$, then $x_+$ is the unique positive real root of $f$, $x_+$ is simple and for any root $z$ of $f$, $x_+\geqslant\,$$|z|$.
%
If $f$ has a root $z\ne x_+$  satisfying $|z|=\,$$x_+$, then
$z$ is simple, and
$z^{\delta}\in \R_{>0}$ for some $\delta\in \N_{>1}$, or equivalently,
the argument of $z$ satisfies $\textup{Arg}(z)/\pi\in\,$$\mathbb{Q}$.
\end{cor}

\begin{exmp}
$f = (x^2 -5)(x^2 -2x + 5) \notin\Phi_1$, because the root
$z=1 + 2\im$ satisfies $|z|=\sqrt{5}$ but $z^\delta\notin\R_{>0}$
for any $\delta\in\N$.
\end{exmp}

The following example shows that the multiplicity for
a root $z$ satisfying $|z| < x_+$ could be any number.
\begin{exmp}
For any $n,k\in\N_{>1}$, $(x+1)^n(x-k)\in\Phi_1$.
Let $n=1$, $(x+1)(x-k)\in\Phi_0$.
Let $f_1(x)=(x+1)^2$ and $f_{n+1}(x)=f_n(x)(x^{2\lfloor\footnotesize{\textup{deg}}(f_n)/2\rfloor+1}+1)$
for $n>1$. Then we have $(x+1)^{n+1}\,|\,f_n(x)$, $f_n(x)\in\Z^{>0}[x]$, and all coefficients of $f_n$ are either $1$ or $2$. Thus, $f_n(x)(x-k)\in\Phi_0$ and $(x+1)^n(x-k)\in\Phi_1$ by definition.
\end{exmp}

\begin{lem}\label{lm-Ru}
Let $q(x)\in\Zx$ be a primitive irreducible polynomial and $\delta\in\N_{>1}$. Then $(q)_{\Zx}\bigcap\Z[x^\delta]=(\widetilde{q}(x^\delta))_{\Z[x^\delta]}$,
where $\widetilde{q}\in\Z[x]$ is primitive and irreducible and
$\widetilde{q}(x^\delta)^m = R_u(u^\delta-x^\delta,q(u))$ for some $m\in\N$.
We use $R_u$ to denote the Sylvester resultant w.r.t. the variable $u$.
Furthermore, the roots of $\widetilde{q}(x)$ are $\{z^\delta\,|\, q(z)=0\}$.
\end{lem}
\begin{proof}
Let $q(x)=a\prod_{j=1}^n(x-z_j)$,  $\zeta_\delta=e^{2\pi \im/\delta}$, and
\begin{eqnarray*}
\overline{R}(x^\delta)
&=& R_u(u^\delta-x^\delta,q(u)) = \prod_{l=1}^\delta q(\zeta_\delta^{l}x).
\end{eqnarray*}
We claim  that $\overline{R}(x^\delta)$ is primitive.
We have $\textup{lc}(R_u(u^\delta-x^\delta,q(u)))=\textup{lc}(\prod_{l=1}^\delta q(\zeta_\delta^{l}x))= a^{\delta}$.
Let $c\in\Z$ be a prime factor of $a^{\delta}$ or $a$.
Since $q$ is primitive, $q\ne0\ (mod\ c)$. Let $q(x)=bx^m+\cdots\ (mod\ c)$.
Then $\lt(\overline{R}(x^\delta))= \textup{lt}(\prod_{l=1}^\delta q(\zeta_\delta^{l}x))=\prod_{l=1}^\delta b(\zeta_\delta^{l}x)^m= b^\delta x^{\delta m}\ne0\ (mod\ c)$. So $c\nmid \overline{R}(x^\delta)$ and thus $\overline{R}(x^\delta)$ is primitive.

Since $\Q[x^\delta]$ is a PID and $\overline{R}(x^\delta)\in(q)_{\Q[x]}\bigcap\Q[x^\delta]$,
there exists a primitive polynomial $\widetilde{q}\in\Zx$
such that $(\widetilde{q}(x^\delta))_{\Q[x^\delta]}=(q)_{\Q[x]}\bigcap\Q[x^\delta]$.
Since $q(x)|\widetilde{q}(x^\delta)$ and $q$ is irreducible,
$\widetilde{q}(x)$ must be irreducible.
%
Since both $q(x)$ and  $\widetilde{q}(x)$ are primitive,
we can deduce
 $(\widetilde{q}(x^\delta))_{\Z[x^\delta]}=(q)_{\Z[x]}\bigcap\Z[x^\delta]$
from  $(\widetilde{q}(x^\delta))_{\Q[x^\delta]}=(q)_{\Q[x]}\bigcap\Q[x^\delta]$.

Since $q(x)|\widetilde{q}(x^\delta)$, $Z_\delta = \{\zeta_\delta^k z_j\,|\, k=1,\ldots,\delta,j=1,\ldots,n\}$ is a subset of the roots of $\widetilde{q}(x^\delta)$.
Let $\overline{S}(x)$ be the square-free part of  $\overline{R}(x)\in\Zx$,
which is also primitive.
Since $Z_\delta$ contains exactly the roots of $\overline{R}(x^\delta)$ and $\overline{S}(x^\delta)$,
we have $\overline{S}(x) | \widetilde{q}(x)$.
Since $\widetilde{q}(x)$ is irreducible and $\overline{S}(x)$ is the square-free part of $\overline{R}(x)$, we have $\overline{S}(x) = \widetilde{q}(x)$
and hence $\overline{R}(x^\delta) = \widetilde{q}(x^\delta)^m$ for some $m\in\N[x]$.
Finally, since the roots of $\widetilde{q}(x^\delta)$ are $\Z_\delta$,
the roots of $\widetilde{q}(x)$ are $\{z^\delta\,|\, q(z)=0\}$.
\end{proof}

\begin{cor}\label{lm-z1}
Let $\delta\in\N$ and  $f=\prod_{j=1}^m q_j^{\alpha_j}$,  where
$\in\N$ and $q_j$ are primitive irreducible polynomials in $\mathbb{Z}[x]$ with positive leading coefficients.
%
Let $q_i^*(x^\delta)$ be the square-free part of $R_u(u^\delta-x^\delta,q_i(u))$  and  $f^* \triangleq \lcm(\{q_j^{*\alpha_j}\,|\, j\})$.
Then
\begin{equation}\label{eq-f*}
(f)_{\Zx}\bigcap\mathbb{Z}[x^\delta] = (f^*(x^\delta))_{\Z[x^\delta]}.
\end{equation}
Furthermore, the roots of $f^*(x)$ are $\{z^\delta\,|\, f(z)=0\}$.
\end{cor}
\begin{proof}
By Lemma \ref{lm-Ru}, we have $(q_i)_{\Zx}\bigcap\mathbb{Z}[x^\delta]=(q_i^*(x^\delta))_{\Z[x^\delta]}$.
Then
$(f)_{\Zx}\bigcap\mathbb{Z}[x^\delta]$
$=\bigcap\limits_{i=0}^s((q_i^{\alpha_i})_{\Zx}\bigcap\mathbb{Z}[x^\delta]
= \bigcap\limits_{i=0}^s(q_i^{*\alpha_i})_{\Z[x^\delta]}
=(\textup{lcm}(\{q_i^{*\alpha_i}\,|\,i\}))_{\Z[x^\delta]} = (f^*(x^\delta))_{\Z[x^\delta]}.$
From $f^* \triangleq \lcm(\{q_j^{*\alpha_j}\,|\, j\})$ and Lemma \ref{lm-Ru},
the roots of $f^*(x)$ are $\{z^\delta\,|\, f(z)=0\}$.
\end{proof}

\begin{thm}\label{th-cr1}
Let  $f\in\,$$\mathbb Z[x]$ have a unique positive root $x_+$
and any root $w$ of $f$ satisfies $|w|\le x_+$.
If there exists a minimal $\delta\in\N_{>1}$ such that
for all root $z\ne x_+$ of $f$, $|z|=x_+$ implies $z^\delta\in\R_{>0}$.
Let $f^*(x^\delta)\in\Z[x^\delta]$ be the polynomial in \bref{eq-f*}.
Then  $f\in\Phi_1$ if and only if  $\lc(f) = \lc(f^*)$ and
$f^* \in \Phi_1$.
\end{thm}
\begin{proof} ``$\Leftarrow$"
Since $\lc(f) = \lc(f^*)$ and $(f)\cap\Z[x^\delta] = (f^*(x^\delta))$,
there exists a monic polynomial $h\in\Z[x]$ such that $f^*(x^\delta) = fh$.
Since $f^* \in \Phi_1$, there exists a monic polynomial $g\in\Zx$ such that $f^*(x) g(x) \in \Phi_0$.
Then $f^*(x^\delta) g(x^\delta) =  fhg(x^\delta)\in \Phi_0$.
Since $hg(x^\delta)$ is monic, we have $f\in \Phi_1$.

``$\Rightarrow$"
Since $f\in\,$$\Phi_1$, there exists a primitive polynomial $h\in\,$$(f)\bigcap \Phi_0$ with $h(0)\neq\,$$0$ and $ \textup{lc}(h)=\,$$\textup{lc}(f)$.
Each such  $h$ has some roots whose absolute value is $x_+$.
Since $f|h$, by part 3 of Lemma \ref{lm-81} we have $\delta| \delta_h$,
where $\delta_h =\gcd\{k\,|\, x^k\hbox{ is in } h \}$.
By Lemma \ref{lm-81}, $h\in\,$$\mathbb{Z}[x^{\delta_h}] \subset \Z[x^{\delta}]$. Thus $h\in(f)\bigcap\mathbb{Z}[x^\delta]=(f^*)_{\Z[x^\delta]}$.
Since $\textup{lc}(f)\,|\,\textup{lc}(f^*)\,|\,\textup{lc}(h)$ and $\textup{lc}(f)=\textup{lc}(h)$,
we have $\textup{lc}(f)=\textup{lc}(f^*)=\textup{lc}(h)$, so $f^*\in \Phi_1$.
\end{proof}

\begin{exmp}\label{ex-51}
Let $f=(x^2-2)(x+1)$. Then $\delta=2$ and $f^* = (x-2)(x-1)$ has two positive roots and
hence $f\not\in\Phi_1$ by Corollary \ref{cor-af8} and Theorem \ref{th-cr1}.

Let  $f_1=x^2-2$, $f_2 = x^2-2x+2$, and $f=f_1f_2$.
Then $\delta = 8$,  $f_1^*=x-16$, $f_2^*=x-16$, and $f^*=x-16$.
Hence $f\in\Phi_1$.
%
\end{exmp}

\begin{cor}\label{cor-z3}
Let $f^*(x)$ be the polynomial defined in Theorem \ref{th-cr1}.
Then $f^*(x)$ has only one root (may be a multiple root) whose absolute value is $x_+^\delta$
and any root $z\ne
x_+^\delta$ of $f^*$ satisfies $|z| < x_+^\delta$.
\end{cor}
\begin{proof}
By Corollary \ref{lm-z1}, the roots of $f^*(x)$ are $\{z^\delta\,|\, f(z)=0\}$.
Then the corollary comes from the fact that $x_+$ is the unique positive real root of $f$
and $f(z)=0, |z|=x_+$ imply $z^\delta\in\R_{>0}$.
\end{proof}

By Corollary \ref{cor-z3}, when $f$ has a unique positive real root $x_+$,
we reduce the decision of $f\in \Phi_1$ into the decision of $f^*\in \Phi_1$,
where $f^*$ has only one root with absolute value $x_+^\delta$.

\begin{lem}
If $f\in\,$$\Phi_1\setminus\Phi_0$ has a unique positive real root $x_+$, then $x_+\geqslant\,$$1$.
\end{lem}
\begin{proof}
There exists a monic polynomial $g\in\,$$\mathbb{Z}[x]$ such that $fg\in\,$$\Phi_0$. Since $f\notin\,$$\Phi_0$, $g$ is not a monomial. Without loss of generality we assume $g(0)\ne0$, and then $\prod_{g(z)=0}|z|=|g(0)/\textup{lc}(g)|=|g(0)|\ge1$
 which implies $\max_{g(z)=0}(|z|)\geqslant\,$$1$.
Since $x_+$ is the unique positive root of $fg$, by Lemma \ref{lm-8}, we have $x_+\geqslant\,$$\max_{g(z)=0}(|z|)\geqslant\,$$1$.
\end{proof}

The following two lemmas give simple criteria to check whether $f\in \Phi_1$
in the case of $f(1)=0$.

\begin{lem}\label{lm-c11}
Let $f\in\,$$\mathbb Z[x]$ be a primitive polynomial, $f(1)=\,$$0$.
If $\delta\in\N$ is the smallest number such that all root $z$ of $f$ satisfies $z^\delta  =1$,
then $f\in\,$$\Phi_1$ if and only if $f^*(x)=\,$$x-1$, where $f^*$ is defined in \bref{eq-f*}.
\end{lem}
\begin{proof}
By Theorem \ref{th-cr1}, if $f^*(x)=\,$$x-1$ then $f\in\Phi_1$.
Suppose $f\in\Phi_1$. By Lemma \ref{lm-81}, any root of $f$ is simple and hence $f$ is square-free.
Let $\delta = \lcm\{m\in\N\,|\, z^m = 1\}$. Since $f$ is primitive,
$\delta\in\N$ is the smallest number such that $f(x)\,|\,x^\delta-1$ in $\Z[x]$.
Therefore, so $f^*(x)\,=\,x-1$.
\end{proof}
\begin{exmp}\label{ex-cr1}
Let $f=(x-1)(x^2+1)(x^3+1)$. Then $\delta=12$ and $f^* = x-1$. So, $f\in\Phi_1$.
Let $f=(x-1)(x^2+1)^2(x^3+1)$. Then $\delta=12$ and $f^* = (x-1)^2$. So, $f\notin\Phi_1$.
\end{exmp}

\begin{lem}\label{lm-c12}
If $f(1)=0$ and any other root $z$ of $f$ satisfies $|z|<\,$$1$, then $f\in\,$$\Phi_1$ if and only if $f(x)/(x-\,$$1)\in\,$$\mathbb{Z}[x^\delta]$ for some $\delta\in\N_{>0}$ and $f(x)(x^\delta-\,$$1)/(x-\,$$1)\in\Phi_0$.
\end{lem}
\begin{proof}
The necessity is obvious.
For the other direction, there exists a monic polynomial $g\in\,$$\mathbb{Z}[x]$ such that $fg\in\,$$\Phi_0$. We claim that each root $z$ of $g$ has absolute value $1$.
Since $g$ is monic,  $\prod_{g(z)=0}|z|\ge1$.
Since  $fg\in \Phi_0$ and $f(1)=0$,  $\max_{g(z)=0}|z| \le1$,
and the claim is proved.

By Lemma \ref{lm-8}, $fg\in\mathbb{Z}[x^{\delta}]$, where $\delta=\delta_{fg}$.
Since $f(1)=0$ and all other roots of $f$ have absolute value $<1$, we have $(x^\delta-1)\,|\,fg$ and $((x^\delta-1)/(x-1))\,|\,g$.
By part 3 of Lemma \ref{lm-81}, the roots of $fg$ with absolute value
$1$ are exactly the roots of $x^{\delta}-1$.
Since the absolute values of all roots of $g$ is $1$ and $g$ has no multiple roots by Lemma \ref{lm-81}, $g= (x^\delta-1)/(x-1)$.
Since $fg\in\mathbb{Z}[x^{\delta}]$ and $(x^\delta-1)\,|\,fg$, set $fg = (x^\delta-1)h(x^\delta)$ for $h\in\Zx$.
From $g= (x^\delta-1)/(x-1)$, we have $f/(x-1) = h(x^\delta)\in\Z[x^\delta]$.
\end{proof}

Now, only  when $f\notin \Phi_0$, $f$ has a unique positive real root $x_+>1$, and any other root of $f$ has absolute value $<x_+$, we
do not know how to decide $f\in \Phi_1$.
By computing many examples, we propose the following conjecture.
\begin{conj}\label{con-1}
If $f\in\,$$\mathbb{Z}[x]\backslash \Phi_0$ has a simple and unique positive real root $x_+$, $x_+ > 1$, and $x_+ > |z|$ for any other root $z$ of $f$, then $f\in \Phi_1$.
\end{conj}

\subsection{Algorithm for $f\in\Phi_1$}

Based on the results proved in the preceding section, we give the following algorithm to decide
whether $f\in\Phi_1$. Note that the last step of the algorithm depends on whether
Conjecture \ref{con-1} is true.

\begin{algorithm}[H]\label{alg-phi1}
  \caption{\bf --- Membership$\Phi_1$ $(f)$} \smallskip
  \Inp{$f\in\Zx$ such that $\lc(f)>0$, $f(0)\ne0$, and $f$ is primitive.}\\
  \Outp{Whether $f\in\Phi_1$.}\medskip

  \noindent
1. If $\lt(f)=f^+$, then $f\in\Phi_0\subset \Phi_1$.\\
2. If $f$ has no positive real roots, then $f\in \Phi_1$.\\
3. If $f$ has at least two positive real roots (with multiplicities counted), then $f\notin \Phi_1$.\\
4. Let $x_+$ be the simple and unique positive real root of $f$.\\
 \SPC 4.1. If $x_+<1$, or equivalently $f(1)>0$ , then $f\notin \Phi_1$.\\
 \SPC 4.2. If $x_+=1$ and all root $z$ of $f$ satisfies $z^\delta=1$ for some $\delta\in\N$,
then  $f\in\Phi_1$ if and only if  $f^* = x-1$, where $f^*$ is defined in \bref{eq-f*}.\\
 \SPC 4.3. If $x_+=1$ and any other root $z$ of $f$ satisfies $|z|<1$, then $f\in\,$$\Phi_1$ if and only if  $f(x)/(x-\,$$1)\in\,$$\mathbb{Z}[x^\delta]$ for some $\delta\in\,$$\N_{>1}$ and $f(x)(x^\delta-\,$$1)/(x-\,$$1)\in\Phi_0$.\\
 \SPC 4.4. If $f$ has a root $z$ such that $|z|>x_+$, then $f\notin \Phi_1$.\\
 \SPC 4.5. If $f$ has a root $z$ such that $z\ne x_+$, $|z|=x_+$,
    and $(\frac{z}{x_+})^\delta\ne1$ for any $\delta\in\N_{>1}$, then $f\notin \Phi_1$.\\
 \SPC 4.6.
Let $\delta$ be the minimal integer such that
$f(z)=0$, $z\ne x_+$, and $|z|=x_+$ imply $(\frac{z}{x_+})^\delta=1$.
Then $f\in \Phi_1$ if and only if $\lc(f)=\lc(f^*)$ and $f^*\in \Phi_1$, where $f^*$ is defined in \bref{eq-f*}.
If $\lc(f)=\lc(f^*)$ then return {\bf{Membership}$\Phi_1$$(f^*)$},
otherwise return false.\\
\SPC 4.7. If $f$ does not satisfy all the above conditions, then it satisfies the condition of Conjecture \ref{con-1} and $f\in \Phi_1$ if  the conjecture is valid.\\
\smallskip
\end{algorithm}

In what below, we will give the details for Algorithm 1
and prove its correctness.
We will use algorithms for real root isolation
and complex root isolation for univariate polynomials.
Please refer to the latest work on these topics and
references in these papers \cite{ri-real,ri-comp}.

Step 1 is trivial to check.
Step 2 can be done with any real root isolation algorithm.
Step 3 can be done by first factoring $f$ as the product of irreducible polynomials
and then isolating the real roots of each factor of $f$.

Step 4.1 is trivial to check.
For Step 4.2, there exists a $\delta\in\N$ such that $(z)^\delta=1$
if and only if each irreducible factor of $f(x)$ is a cyclotomic polynomial,
which can be checked with the Graeffe method in \cite{Bradford}
and the $\delta$ can also be founded.
The polynomial $f^*$ in Step 4.2 can be computed with Corollary \ref{lm-z1}.

In Step 4.3, the $\delta$ can be found from the fact $f(x)/(x-1)\in\Z[x^\delta]$.
If $f(x)(x^\delta-1)/(x-1)\in\Phi_0$ for some $\delta$ satisfying $f(x)/(x-1)\in\Z[x^\delta]$, then return true; otherwise return false.

In Steps 4.4, 4.5, and 4.6, we need to check whether
$f$ has a root $z\ne x_+$ such that  $|z|>x_+$,
$|z|=x_+$, and $z^m\in\R_{>0}$ for some $m\in\N$.
%
%
To do that, we first give a lemma.
\begin{lem}\label{lm-zb}
Let $p(x)=a\prod_{i=1}^n(x-x_i)\in\Zx$, $q(x)=b\prod_{j=1}^m(x-y_j)\in\Zx$,
and $x_iy_j\ne0$ for all $i,j$.
Then the roots of $R_u(p(u),q(ux))$ are $\{y_j/x_i\,|\,i=1,\cdots,n, j=1,\cdots,m\}$
and the roots of $R_u(u^n p(x/u),q(u))$ are $\{x_iy_j\,|\,i=1,\cdots,n, j=1,\cdots,m\}$.
\end{lem}
\begin{proof}
The lemma comes from
$R_u(p(u),q(ux))=a^m b^n\prod_{i,j}(x-x_j/y_i)$
and
$R_u(u^np(x/u),\,q(u))=a_0^{m}b^n\prod_{i,j}(x-x_ix_j)$,
where $a_0 = p(0)$.
\end{proof}

In the rest of this section, we assume
 \begin{eqnarray}
 f&=&f_0 \prod_{i=1}^t f_i^{e_i}\cr
 r_i(x) &=& R_u(u^n f_i(x/u),f_i(u))
\end{eqnarray}
where $f_i$ are primitive and irreducible polynomials with positive leading coefficients.
Also assume that $f(x)$ has a unique positive root $x_+$ which is the root of $f_0(x)$.

By Lemma \ref{lm-zb}, the real roots of all $r_i(x)$ include $x_+^2$
and $z\overline{z}$, where $z$ is a complex root of $r_i(x)$.
Then the condition in Step 4.4 of the algorithm can be checked with the following result
based on real root isolation.
\begin{cor}\label{cor-a1}
$f$ has a root $z$ such that $|z|>x_+$ if and only if
some $r_i(x)$ has a positive root larger than $x_+^2$.
\end{cor}

It is easy to check whether $-x_+$ is a root of $f_i$:
since $f_i$ is irreducible, $-x_+$ is a root of $f_i$
if and only if $f_i(-x) = \pm f_i(x)$.
If $z$ is complex root of $f_i$ such that $|z|=x_+$,
then $x_+^2, x_+^2=z.\overline{z}, x_+^2=\overline{z}.z$ are all roots of $r_i$.
Then, we have the following result.
\begin{cor}\label{cor-a2}
Let $m_i$ be the multiplicity of $x_+^2$ as a root of $r_i$
and $n_i$ the multiplicity of $-x_+$ as a root of $f_i$
(the multiplicity is set to be zero if $x_+^2$ or  $-x_+$ is not a root).
Then $\#\{z\,|\,f_0(z)=0, |z|=x_+, z\notin\R\} = m_0-n_0-1$ and
$\#\{z\,|\,f_i(z)=0, |z|=x_+,z\notin\R\} = m_i-n_i$ for $i>0$.
\end{cor}

As usual, a {\em representation} of a complex root $z$ is a pair
$(p,B)$ where $p$ is an irreducible polynomial and
$B$ a box such that $p(z)=0$ and $z$ is the only root of $p$ in $B$.
A box is represented by its lower-left and upper-right
vertexes: $([x_l,y_l],[x_t,x_t])$.
By the following lemma, we can find representations for
all roots $z$ of $f$ satisfying $|z|=x_+$.
\begin{lem}\label{lm-a2}
Suppose $f_i$ has $s$ roots $z_1,\ldots,z_s$ satisfying $|z_j|=x_+$.
Then, we can find representations for $z_j$.
\end{lem}
\begin{proof}
Since $f_i$ is irreducible, $f_i$ is the minimal polynomial for $z_i$.
Suppose $I=(a,b)$ is an isolation interval for $x_+$.
By algorithms of complex root isolation and real root isolation,
we can simultaneously refine $I$ and the
isolation boxes of the roots of $f_i$
such that the number of isolation boxes
meet the region $a<|x|<b$ will eventually becomes $s$.
These $s$ boxes are the isolation boxes for $z_1,\ldots,z_s$,
since $f_i$ has exactly $s$ roots satisfying $|z|=x_+$.
\end{proof}

\begin{lem}
Let $z$ be a root of $f_k$ satisfying $|z|=x_+$.
Then, we can find a representation for $z/x_+$.
\end{lem}
\begin{proof}
Let $H(x)=R_u(f_0(u),f_k(ux))\in\Zx$ and $h_i(x),i=1,\ldots,s$ the irreducible factors of $H$.
From Lemma \ref{lm-zb}, $H(z/x_+)=0$ and $h_c(z/x_+)=0$ for certain $c$
and we will show how to find $h_c$.
Isolate the roots of $h_i,i=1,\ldots,s$ and refine the isolation box $B=([x_l,y_l],[x_t,x_t])$ of $z$
and the isolation interval of $x_+=(l,r)$ simultaneously such that
$([x_l/r,y_l/r],[x_t/l,x_t/l])$ intersects only one of the isolation boxes of $h_i,i=1,\ldots,s$.
This box $B_1$ should be the isolation box for $z/x_+$.
If $B_1$ contains a root of $f_c$, then $f_c$ is the minimal polynomial for  $z/x_+$.
\end{proof}

With the following lemma, we can check whether  $z^m\in\R_{>0}$ for some $m$.
\begin{lem}
Let $z$ be a root of $f_k$ satisfying $|z|=x_+$
and $q$ the minimal polynomial for $z/x_+$.
Then we can decide whether there exists  an $m\in\N$ such that $(z/x_+)^m=1$,
and if such an $m$ exists, we can compute the minimal $m$.
\end{lem}
\begin{proof}
There exists an $m\in\N$ such that $(z/x_+)^m=1$
if and only if $q(x)$ is a cyclotomic polynomial,
which we can be tested by the Graeffe method in \cite{Bradford}.
The method also gives the $m$ such that $(z/x_+)^m=1$.
The minimal $m$ can be found easily.
\end{proof}

Now, we consider Step 4.5. With Corollary \ref{cor-a2} and Lemma \ref{lm-a2},
we can find all the roots
$z$ of $f$ satisfying $|z|=x_+$. For each such $z$,
we can check whether there exists a $\delta_z\in\N$
such that $(z/x_+)^{\delta_z}=1$ with Lemma \ref{lm-cyc}.
Hence the conditions of Step 4.5 can be checked.

Now, we consider Step 4.6.
The $\delta$ in Step 4.6 can be computed as  $\delta = \lcm\{\delta_z\,|\,
f(z)=0, |z|=x_+, (z/x_+)^{\delta_{z}}=1\}$.
With $\delta$ given, $f^*$ in Step 4.6 can be computed with Corollary \ref{lm-z1}.
From Corollary \ref{lm-z1}, the roots of $f^*$ are $\{z^\delta\,|\, f(z)=0\}$.
%
%
As a consequence, when running {\bf{Membership$\Phi_1$}$(f^*)$},
only Steps 1, 3, 4.7 will be executed,
and no further calls to {\bf{Membership$\Phi_1$}$(f^*)$} are needed.

\subsection{Compute the finite $\sigma$-Gr\"obner basis}

Let $f\in\Phi_1$, we will show how to compute the finite $\sigma$-Gr\"obner basis
for $\ID_f = \sat(\PH_f)$ in \bref{eq-idf1}.
\begin{lem}
Let $f\in\Phi_1$, $h=fg\in\Phi_0$ for a monic $g\in\Zx$,
and $D=\deg(h)$. Then
\begin{equation}\label{eq-idd}
\ID_D=\sat(\PH_f)\bigcap \F[y,y^x,\cdots,y^{x^D}] =
\asat(\PH_f,\PH_{xf},\ldots,\PH_{x^{D-{\tiny\deg}(f)}f})
\end{equation}
and a Gr\"obner basis of $\ID_D$ is a $\sigma$-Gr\"obner basis of $\ID_f$.
\end{lem}
\begin{proof}
By the remark before Theorem \ref{th-nbgb1}, $\PH_f$ is regular and coherent.
Then $P\in \ID_D$ if and only if $\prem(P,\PH_f)=0$
which is equivalent to $P\in \asat(\PH_f,\PH_{xf},\ldots,\PH_{x^{D-{\tiny\deg}(f)}f})$
\cite{dbi}, and \bref{eq-idd} is proved.
By Corollary \ref{cor-s2}, a Gr\"obner basis of $\ID_D$ is a $\sigma$-Gr\"obner basis of $\ID_f$.
\end{proof}

The Gr\"obner basis of $\ID_D$,  denoted as $\GB(f,D)$,
can be computed with the following well-known fact
$$\asat(\PH_f,\PH_{xf},\ldots,\PH_{x^{D-{\tiny\deg}(f)}f}) =
(z\cdot J^{\tiny\sum_{i=0}^{D-deg(f)} x^i}-1,\PH_f,\PH_{xf},\ldots,\PH_{x^{D-{\tiny\deg}(f)}f})
\cap\F[y,y^x,\cdots,y^{x^D}],$$
where $J = \init(\PH_f)$ and $z$ is a new indeterminate.
Therefore, in order to compute the $\sigma$-Gr\"obner basis of $\ID_f$,
it suffices to compute $D$. We thus have the following algorithm.

\begin{algorithm}[H]\label{alg-gb}
  \caption{\bf --- FiniteGB $(f)$} \smallskip
  \Inp{$f\in\Phi_1$ such that $\lc(f)>0$.}\\
  \Outp{Return $\sigma$-Gr\"obner basis of $\ID_f=\sat(\PH_f)$.}\medskip

  \noindent
 1. If $\lt(f)=f^+$, then return $\{\PH_f\}$.\\
 2. If $f$ has no positive real roots, then return $\GB(f,N_f+\deg(f)+1)$,
   where $N_f$ is defined in Lemma \ref{lm-Po}.\\
3. Let $x_+$ be the unique simple  positive real root of $f$.\\
\SPC 3.1.
If $x_+=1$ and all root $z$ of $f$ satisfies $z^\delta=1$ for some $\delta\in\N$,
then return $\GB(f,\delta)$.\\
\SPC 3.2. If $x_+=1$ and any other root $z$ of $f$ satisfies $|z|<1$, then return $\GB(f,\deg(f)+\delta-1)$, where $\delta$ is found in Step 4.3 of Algorithm 1.\\
\SPC 3.3.
Let $\delta$ be the minimal integer such that
$f(z)=0$, $z\ne x_+$, and $|z|=x_+$ imply $(\frac{z}{x_+})^\delta=1$.
Let the $f^*$ be defined \bref{eq-f*} and  $f^*(x^\delta) = f(x) s(x)$.
Return $\GB(f,\delta\deg(f^*))$.\\
%
%
\smallskip
\end{algorithm}

In the rest of this section, we will prove the correctness of the algorithm.
Step 1 follows Lemma \ref{lm-s1}.
%

For Step 2,  by Lemma \ref{lm-Po},
$(x+1)^{N_f} f\in\Z^{>0}[x]$.
Following the proof of Lemma \ref{lm-s3},
for a sufficiently large $M\in\N$,
$(x-M) (x+1)^{N_f} f\in\Phi_0$.
Then, $D = \deg((x-M) (x+1)^{N_f} f)=N_f + \deg(f) + 1$.

For Step 3.1, following Step 4.2 of Algorithm 1,
we have $f^*(x^\delta) = f(x) g(x) = x^\delta -1$ for some $g$.
Then $D = \delta$.
For Step 3.2, following Step 4.3 of Algorithm 1,
$f(x)(x^\delta -1)/(x-1)\in\Phi_0$.
Then $D =\deg(f)+ \delta-1$.

For Step 3.3, from the proof of Step 4.6 of Algorithm 1,
there exist three possibilities:
$f^*(x)\in\Phi_0$, $f^*(x)$ has at least two positive roots,
or $f^*$ satisfies the conditions of Conjecture \ref{con-1}.
Since we already assumed $f^*\in\Phi_1$, only  $f^*(x)\in\Phi_0$ is possible.
From $f^*(x^\delta) = f(x) s(x)$, we have $D= \delta\deg(f)$.
%
%
We now proved the correctness of Algorithm 2.

\section{Approach based on integer programming and lower bound}
Given an $f\in\Zx$, the existence of a monic polynomial
$g\in\mathbb{Z}[x]$ with $\deg (g)\le  m$, such that $fg\in\Phi_0$
can be reduced to an integer programming problem.
Based on this idea, a lower bound for $\deg(g)$ is given in certain cases.

\begin{lem}\label{lm-ip}
Given a polynomial $f(x)=a_nx^n+\cdots+a_0\in\mathbb{Z}[x]$ with $a_n>0$, there exists a monic polynomial $g\in\mathbb{Z}[x]$ with $\deg (g)\le  m$, such that $fg\in\Phi_0$ if and only if a $(b_{m-1},\cdots,b_0)\in\mathbb Z^m\,$ satisfies
\begin{eqnarray}\label{mx-0}
\left(
  \begin{array}{ccccccc}
    {a_{n-1}} & {a_n} & & & \\
    \vdots& \vdots& \ddots &  \\
    {a_0} & {a_{1}} &\cdots  & {a_{n}}&  \\
       &\ddots &\ddots & \ddots& \ddots \\
     & & {a_0} &a_1&\cdots  &{a_{n}}   \\
    &  && \ddots&\ddots &\vdots\\
     &   &  &&a_0 & a_1\\
     &&  & & & {a_0} \\
  \end{array}
\right)_{(m+n)\times(m+1)}
\left(
  \begin{array}{c}
    {1} \\
    {b_{m-1}} \\
    {b_{m-2}} \\
    \vdots \\
    {b_0} \\
  \end{array}
\right)\le0.
\end{eqnarray}
Moreover such $g$ has degree $<m$ if and only if  $b_0=0$ for some feasible solution of the above inequalities.
\end{lem}
\begin{proof}
Let $g(x) = x^m+b_{m-1}x^{m-1}+\cdots+b_0$. The leading coefficient of $fg$ is $a_n>0$, and the coefficient of $x^k$ is the $m+n-k$-row of the left side of \bref{mx-0} for $k=m+n-1,\ldots,0$.
If $\deg(g)<m$, the coefficients of $g_1(x)=x^{m-\tiny{\deg}(g)}g(x)$ is a feasible solution with $b_0=0$. If $b_0=0$, $(1,b_{m-1},\cdots,b_1)$ is a feasible solution of \bref{mx-0} for $m=m-1$.
\end{proof}

The following result gives another criterion for the existence of $g$.
\begin{lem}\label{lm-1f}
Given a polynomial $f(x) = a_nx^n+\cdots+a_0\in\mathbb{Z}[x]$ with $a_n>0$, let $(1/f)(x)\triangleq \lambda_0+\dots+\lambda_mx^m+\dots\in\mathbb Z[a_0^{-1}][[x]]$.
There exists a monic polynomial $g\in\mathbb{Z}[x]$ with $\deg (g)\le m$ and $ fg\in\Phi_0$ if and only if there exists a $(c_{m+n-1},\cdots,c_0)\in\mathbb N^{m+n}\,$ such that
\begin{equation}\label{eq-1f}
\left(
  \begin{array}{cccccc}
    {\lambda_{0}} & {\lambda_{1}} & \lambda_{2} &\cdots & {\lambda_{m+n-2}} &{\lambda_{m+n-1}}\\
    {}& {\lambda_{0}} & {\lambda_{1}}   &\ddots &\ddots & {\lambda_{m+n-2}} \\
   {}& {}& \ddots &\ddots &\ddots & {\vdots}  \\
    {} & {}&{}& {\lambda_{0}} &{\cdots}& {\lambda_{m}} \\
  \end{array}
\right)
\left(
  \begin{array}{c}
    {c_{m+n-1}} \\
    {c_{m+n-2}} \\
    \vdots \\
    {c_0} \\
  \end{array}
\right)=\left(
\begin{array}{c}
{0}\\
 \vdots\\
{0}  \\
-1
\end{array}
\right).
\end{equation}
\end{lem}
\begin{proof}
Extending the proof of Lemma \ref{lm-ip}, let $\mathbf b_{m+n-1}\triangleq (b_{m+n-1},\cdots,b_0)^T$. For the following special Jordan form
$$J_j\triangleq
\left(\begin{array}{cccc}
0&1& &\\
 &\ddots&\ddots& \\
 & &\ddots&1\\
 & & &0\\
\end{array}\right)_{j\times j},\ \textrm{ we have }f(J_j)=
\left(\begin{array}{ccccc}
a_0&\cdots&a_n&&\\
&a_0&\ddots&\ddots&\\
&&\ddots&\ddots&a_n\\
&&&a_0&\vdots\\
&&&&a_0\\
\end{array}\right)_{j\times j}.$$
By Lemma \ref{lm-ip},  $fg\in\Phi_0$ if and only if $f(J_{m+n})\,\mathbf{b}\in\mathbb Z_{\le  0}^{m+n}$ for some $(b_{m-1},\cdots,b_0)\in\mathbb Z^m$ with $(b_{m+n-1},\cdots,b_m)=\,(0,\cdots,0,1)$. Let $\mathbf{c} = (c_{m+n-1},\cdots,c_0)^T\triangleq -f(J_{m+n})\,\mathbf{b}\in\mathbb N^{m+n}$. Then we have $f(J_{m+n})^{-1}\mathbf{c} = (1/f)(J_{m+n})\mathbf{c} = -\mathbf{b}$, that is
 $$\left(\begin{array}{cccc}
    {\lambda_{0}} &{\lambda_{1}} &\cdots &{\lambda_{m+n-1}}\\
    {} &{\lambda_{0}} &\ddots& \vdots \\
    {} &{} &\ddots& {\lambda_{1}} \\
    {} & {}& {}& {\lambda_{0}}\\
  \end{array}
\right)
\left(
  \begin{array}{c}
    {c_{m+n-1}} \\
    {c_{m+n-2}} \\
    \vdots \\
    {c_0} \\
  \end{array}
\right)=\left(
\begin{array}{c}
{0}\\
 \vdots\\
{0}  \\
-1\\
-b_{m-1}\\
\vdots\\
-b_0\\
\end{array}
\right).$$
Since we need only to know the existence of $c_i$, only the first $n$ rows
are need, and the lemma is proved.
\end{proof}

Note that $a_0^{i+1}\lambda_i\in\Z$ for any $i\in\N$.
We can reduce the coefficient matrix in the above lemma into
an integer matrix.
\begin{cor}\label{cor-lbq}
Let  $f,g\in\R[x]$, $\lc(f)>0$, $g$ monic, and
$(1/f)(x)\triangleq \sum_{m=0}^{\infty}\lambda_mx^m \in\R[[x]]$.
If $\lt(fg) = (fg)^+$, then $\deg (g)\ge \min\{j\in\N\,|\,\lambda_j<0\}$.
\end{cor}
\begin{proof}
From the proof of Lemma \ref{lm-1f},
there exists a monic $g\in\R[x]$ such that $\lt(fg) = (fg)^+$
if and only if  \bref{eq-1f} has a solution $(c_{m+n-1},\cdots,c_0)\in\R_{>0}^{m+n}$.
If $\lambda_0,\dots,\lambda_m\ge0$, the last coordinate of \bref{eq-1f} is $\sum_{j=0}^m\lambda_jc_{m-j}\ge0$, hence $\sum_{j=0}^m\lambda_jc_{m-j}\ne-1$
and  \bref{eq-1f} has no solution in $\R_{>0}^{m+n}$.
As a consequence, if $\lt(fg) = (fg)^+$, then $\deg (g)\ge \min\{j\in\N\,|\,\lambda_j<0\}$
and the corollary is proved.
\end{proof}

\begin{cor}\label{cor-lbq2}
Let $f(x) = ax^2+bx+c\in\R[x]$,  $a>0$, $b^2-4ac<0$, and $z$ a root of $f$.
If $fg\in\Phi_0$ and $g$ is monic, then $\deg (g)\ge \lfloor\pi/|\textup{Arg}(z)|\rfloor = \lfloor\pi/\arctan(\sqrt{4ac-b^2}/b)\rfloor$.
\end{cor}
\begin{proof}
Let $f(x)=a(x-z)(x-\bar z)$, and $z=re^{\theta \im}$ where $r\in\R_{>0}$
and $\theta=\textup{Arg}(z)\ne k\pi$.
Without loss of generality, we can assume $0<\theta<\pi$. Then
$$\frac{1}{f(x)}=\frac{1}{a(x-z)(x-\bar z)}=\sum_{j=0}^{\infty}\frac{z^{j+1}-\bar z^{j+1}}{a(z\bar z)^{j+1}(z-\bar z)}x^j=\sum_{j=0}^{\infty}\frac{\sin((j+1)\theta)}{ar^{j+2}\sin\theta}x^j,$$
that is, $\lambda_j=\frac{\sin((j+1)\theta)}{ar^{j+2}\sin\theta}$.
Since $\lambda_0=\frac{1}{ar^{2}}>0$,
$\min\{j\in\N\,|\,\lambda_j<0\}=\min\{j\in\N\,|\, (j+1)\theta>\pi\}=\lfloor\pi/\theta-1\rfloor+1=\lfloor\pi/\theta\rfloor$.
 By Corollary \ref{cor-lbq}, $\deg (g)\ge\lfloor\pi/\theta\rfloor=\lfloor\pi/\arctan(\sqrt{4ac-b^2}/b)\rfloor$.
\end{proof}

We can now give a lower bound for the degree of $g$ such that $fg\in\Phi_0$
in certain case.
\begin{thm}
If a polynomial $f(x)\in \Zx$ is of degree $n$ and has at least one root
not in $\R$, then \\
$\min\{\deg (g)\,|\,g\in\Zx\,\, \hbox{\rm is monic and }  fg\in \Phi_0\} \ge
\max\{ \lfloor\pi/|\textup{Arg}(z)|\rfloor-n+2\,|\, f(z)=0,z\notin\R\}$.
\end{thm}
\begin{proof}
Since $f(x)\in \Zx$ has at least one root not in $\R$, $f = f_1f_2$ where $f_2$
is a quadratic polynomial in $\R[x]$ which has two complex roots.
Suppose there exists a monic $g\in\R[x]$ such that
$\lt(fg) = (fg)^+$ or $\lt(f_1f_2g) = (f_1f_2g)^+$.
By Corollary \ref{cor-lbq2}, $\deg(g) \ge \lfloor\pi/|\textup{Arg}(z)|\rfloor -\deg(f_1) = \lfloor\pi/|\textup{Arg}(z)|\rfloor-n+2$.
Then, $\min\{\deg (g)\,|\,g\in\Zx\,\, \hbox{\rm is monic and }  fg\in \Phi_0\} \ge
\min\{\deg (g)\,|\,g\in\R[x]\,\, \hbox{\rm is monic and }  \lt(fg)=(fg)^+\} \ge
\max\{ \lfloor\pi/|\textup{Arg}(z)|\rfloor-n+2\,|\, f(z)=0,z\notin\R\}$.
\end{proof}

The following result shows that the lower bound given in the preceding theorem
is also the upper bound for quadratic polynomials.
\begin{prop}\label{lm-qu}
Let $f(x)=a_2x^2+a_1x+a_0=a_2(x-z)(x-\bar z)$ be a quadratic polynomial in $\Z[x]$
with a root complex $z=a+b\im=re^{\theta \im}$, where $a_2,b,r>0$, $0<\theta<\pi$, $\bar z=a-b\im$. Then $\min\{\deg (g)\,|\,g\in\Zx \hbox{ and monic},\  fg\in \Phi_0\}=\lfloor\pi/\theta\rfloor.$
\end{prop}
\begin{proof}
%
If $\pi/2<\theta<\pi$, then $a_1 = -2 a>0$ and hence $f\in\N^{>0}[x]$.
By the proof of Lemma \ref{lm-s3}, there exists an $N$ such that $(x-N)f\in\Phi_0$
and hence $\widehat{\deg}(f)=1=\lfloor\pi/\theta\rfloor$.
If $\theta=\pi/2$, then $f=a_2x^2+a_0$. It is easy to check $\widehat{\deg}(f)=2=\lfloor\pi/\theta\rfloor$.

From now on, we assume $0<\theta<\pi/2$, so $a>0$ and $a_1<0$.
Considering $f_1(x)=(x-a-bi)(x-a+bi)=x^2-2ax+a^2+b^2\in\Z[a_2^{-1}][x]$, we will solve the integer programming mentioned in Lemma \ref{lm-ip}:
\begin{eqnarray}\label{mx-1}
\left(
   \begin{array}{ccccc}
     {-2a} & {1} & {} & {} & {} \\
     {a^2+b^2} & {-2a} & {1} & {} & {} \\
     {} &  {\ddots} & {\ddots} & \ddots & {} \\
     {} & {} & {a^2+b^2}&-2a & {1} \\
     {} & {} & {} & a^2+b^2& {-2a} \\
     {} & {} & {} & {} &  {a^2+b^2}\\
   \end{array}\right)_{(m+2)\times(m+1)}\left(
  \begin{array}{c}
    {1} \\
    {b_{m-1}} \\
    {b_{m-2}} \\
    \vdots \\
    {b_0} \\
  \end{array}
\right)\le0.
\end{eqnarray}
Let $\Delta_1 = -2a$ and r $\Delta_{j+1} = {-2a}-({a^2+b^2})/{\Delta_{j}}$ for $j>1$.
Then
$$\Delta_{j}=-\frac{(a+bi)^{j+1}-(a-bi)^{j+1}}{(a+bi)^{j}-(a-bi)^{j}}=-\frac{r\sin(j+1)\theta}{\sin j\theta}.$$
Let $m_0=\lceil\pi/\theta\rceil-1$. Then  we have $\Delta_j<0$ for $j=1,\cdots,m_0-1$ but $\Delta_{m_0}\ge0$.

We will do row transformations on \bref{mx-1} to relax its feasible region.
Let $m=m_0-1$. We add $(m+1)$-th row multiplied by $1/({-\Delta_1})>0$ to the $m$-th row. Then the $-2a$ at the $m$-th row becomes $\Delta_2 = -2a-({a^2+b^2})/{\Delta_1}$, and the $1$ at the $m$-th row becomes $0$.
Then add $m$-th row multiplied by $1/({-\Delta_2})>0$ to the $(m-1)$-th row.
Repeat the above process until $\Delta_{m_0}\ge 0$, and we obtain a lower triangular matrix:
\begin{eqnarray}\label{mx-lt}
\left(
   \begin{array}{ccccc}
    \Delta_{m_0}  & {0} & {} & {} & {} \\
     {a^2+b^2} & {\Delta_{m_0-1}} & {0} & {} & {} \\
     {} &  {\ddots} & {\ddots} & \ddots & {} \\
     {} & {} & {a^2+b^2}&\Delta_2 & {0} \\
     {} & {} & {} & a^2+b^2& {\Delta_1} \\
     {} & {} & {} & {} &  {a^2+b^2}\\
   \end{array}\right)_{m_0+1 \times m_0}.
\end{eqnarray}
\begin{enumerate}
\item
If $\Delta_{m_0}>0$,   the first coordinate of the left side of
\begin{eqnarray}\label{mx-lt1}
\left(
   \begin{array}{cccc}
    \Delta_{m_0}  & {} & {} & {}  \\
     {a^2+b^2} & {\Delta_{m_0-1}} & {} & {}  \\
     {} &  {\ddots} & {\ddots} &   \\
     {} & {} & {a^2+b^2}&\Delta_1 \\
     {} & {} & {} & a^2+b^2 \\
   \end{array}\right)\left(
  \begin{array}{c}
    {1} \\
    {b_{m_0-2}} \\
    {b_{m_0-3}} \\
    \vdots \\
    {b_0} \\
  \end{array}
\right)\le0\end{eqnarray}
is $\Delta_{m_0}>0$. So the feasible region of \bref{mx-lt1} is empty and hence the feasible region of \bref{mx-1} is also empty.
Thus $fg\notin \Phi_0$ for any monic polynomial $g$ of degree $< m_0$ by Lemma \ref{lm-ip}.

Let $m=m_0$. We have
\begin{eqnarray}\label{mx-lt2}
\left(
   \begin{array}{ccccc}
     -2a&1  & {} & {} & {}  \\
  {a^2+b^2}   & \Delta_{m_0}  & {} & {} & {}  \\
    & {a^2+b^2} & {\Delta_{m_0-1}} & {} & {}  \\
   &  {} &  {\ddots} & {\ddots} &   \\
   &  {} & {} & {a^2+b^2}&\Delta_1 \\
   &  {} & {} & {} & a^2+b^2 \\
   \end{array}\right)\left(
  \begin{array}{c}
    {1} \\
    {b_{m_0-1}} \\
    {b_{m_0-2}} \\
    \vdots \\
    {b_0} \\
  \end{array}
\right)\le0.\end{eqnarray}
Similarly, we can obtain a quasi-upper trangular matrix from \bref{mx-1} by row transformations:
\begin{eqnarray}\label{mx-ut}
\left(
      \begin{array}{ccccc}
     {\Delta_1} & {1} & {} & {} & {} \\
     {} & {\ddots} & {\ddots} & {} & {} \\
     {} & {} & {\Delta_{m_0-1}} & 1 & {} \\
     {} & {} & {}& \Delta_{m_0}& {1} \\
     {} & {} & {} & a^2+b^2 & {-2a} \\
   \end{array}\right)\left(
  \begin{array}{c}
    {1} \\
    {b_{m_0-1}} \\
    {b_{m_0-2}} \\
    \vdots \\
    {b_0}\\
    \end{array}\right)\le0.
\end{eqnarray}
Combining \bref{mx-1}, \bref{mx-lt2} and \bref{mx-ut}, we have
\begin{eqnarray}\label{ie-m1}
b_{m_0-1}\le-\frac{a^2+b^2}{\Delta_{m_0}}\Rightarrow b_{m_0-1}<0\Rightarrow-2a+b_{m_0-1}\le0;
\end{eqnarray}
\begin{eqnarray}\label{ie-in}
-\frac{a^2+b^2}{\Delta_{j+1}}b_{j+1}\le b_j\le-(a^2+b^2)b_{j+2}+2ab_{j+1},\ j=m_0-2,m_0-3,\cdots,0;
\end{eqnarray}
\begin{eqnarray}
b_j\le-(a^2+b^2)b_{j+2}+2ab_{j+1}\Rightarrow b_j\le-\Delta_{m_0-j}b_{j+1}\Rightarrow b_j<0, \ j=m_0-2,m_0-3,\cdots,1;
\end{eqnarray}
\begin{eqnarray}
b_0\le0\Rightarrow\Delta_{m_0}b_1+b_0\le0.
\end{eqnarray}
In \bref{ie-in}, we need to show that there exists a rational number $b_j$ satisfying
\begin{eqnarray}\label{ie-in2}
-\frac{a^2+b^2}{\Delta_{j+1}}b_{j+1}<b_j<-(a^2+b^2)b_{j+2}+2ab_{j+1}.\end{eqnarray}
We need to show
$$-(a^2+b^2)b_{j+2}+2ab_{j+1}+\frac{a^2+b^2}{\Delta_{j+1}}b_{j+1}=-(a^2+b^2)b_{j+2}-\Delta_{j+2}b_{j+1}>0,$$
which is true from the first `$<$' in \bref{ie-in2} when $j=j+1$.

Then we can choose some rational number $b_{m_0-1},\cdots,b_0$  satisfying \bref{ie-m1} and \bref{ie-in2}, and then $(1,b_{m_0-1},\cdots,$ $b_0)$ is a feasible solution of \bref{mx-1}. Taking the common denominator $N\in\N_{\ge1}$ of $\{b_j\,|\,j=0,\cdots,m_0-1\}$, we have
$$-2a+Nb_{m_0-1}<-2a+b_{m_0-1}\le0;$$
$$a^2+b^2-2aNb_{m_0-1}+Nb_{m_0-2}<N(a^2+b^2-2ab_{m_0-1}+b_{m_0-2})\le0;$$
$$(a^2+b^2)Nb_j-2aNb_{j-1}+Nb_{j-2}\le0,\ j=m_0-1,\cdots,2;$$
$$(a^2+b^2)Nb_1-2aNb_0\le0;\ (a^2+b^2)Nb_0\le0,$$
and then
\begin{eqnarray}\label{ie-nb}
f(x)g_1(x)=a_2(x^2-2ax+a^2+b^2)(x^{m_0}+\sum_{j=0}^{m_0-1} Nb_jx^{j})\in\Phi_0.
\end{eqnarray}
Then $\Delta_{m_0}>0$ implies $\widehat{\deg}(f)=m_0=\lceil\pi/\theta\rceil-1=\lfloor\pi/\theta\rfloor$.
\item
If $\Delta_{m_0}=0$, $\pi/\theta=m_0+1>2$, $z=re^{\pi \im/({m_0+1})}$. Then $e^{2\pi \im/({m_0+1})}$ is a root of $(x-1)^{-2}R_u(f(x),f(ux))=a_2a_0x^2+(2a_2a_0-a_1^2)x+a_2a_0$. Since $e^{2\pi \im/({m_0+1})}$ is integral over $\Z$, we have $a_0a_2\,|\,(2a_2a_0-a_1^2)$ or $a_0a_2\,|\,a_1^2$.
For $0<2\pi/({m_0+1})<\pi$, $a_2a_0x^2+(2a_2a_0-a_1^2)x+a_2a_0$ has no real roots,
and then we have $(2a_2a_0-a_1^2)^2-4(a_2a_0)^2<0$, that is ${a_1^2}<4{a_0a_2}$. Then we have $m_0=2$ when ${a_1^2}={a_0a_2}$,  $m_0=3$ when ${a_1^2}=2{a_0a_2}$ or $m_0=5$ when ${a_1^2}=3{a_0a_2}$.
\begin{enumerate}
\item
If $m_0=2$ and $\Delta_{2}=0$, $f(x)=a_2x^2+a_1x+a_0$, where $a_1=-\sqrt{a_0a_2}$. When solving \bref{mx-1} for $m=3$, we have
$$b_0\le\frac{a_0^3}{a_1^3},\ b_1\le-\frac{a_1b_0}{a_0},\ -\frac{a_0^2 + a_1^2 b_1}{a_0 a_1} \le b_2\le-\frac{a_1^2 b_0 + a_0 a_1 b_1}{a_0^2}.$$
In order for  an integer $b_2$ to satisfy these inequations, we need to assume
$$\frac{a_1^2 b_0 - a_0 a_1 b_1}{a_0^2}+\frac{a_0^2 + a_1^2 b_1}{a_0 a_1} \ge2,\textup{
that is } \ b_0\le\frac{a_0^3-2 a_0^2 a_1}{a_1^3}.$$
Here $b_0<0$ implies $\min\{\deg(g)\,|\,fg\in \Phi_0\}\ge3$, so $\widehat{\deg}(f)=3=\pi/\theta=\lfloor\pi/\theta\rfloor$.
\item
If $m_0=3$ and $\Delta_{3}=0$, $f(x)=a_2x^2+a_1x+a_0$, where $a_1=-\sqrt{2a_0a_2}$. When we solve \bref{mx-1} for $m=4$, we have
$$b_0 \le\frac{ - a_0^2 }{a_2^2},\
 b_3 \le \frac{a_0^2 a_1 + a_1 a_2^2 b_0}{-a_0^2 a_2},$$
 $$ \frac{-a_2^2 b_0 + a_0 a_1 b_3}{-a_0 a_2}\le b_2 \le\frac{-a_0 - a_1 b_3}{a_2} ,\  \frac{-a_2 b_0 - a_0 b_2}{a_1} \le b_1 \le \frac{-a_1 b_2 - a_0 b_3}{a_2} .$$
 When we want
 $$\frac{-a_1 b_2 - a_0 b_3}{a_2}-\frac{-a_2 b_0 - a_0 b_2}{a_1}\ge2,\ \frac{-a_0 - a_1 b_3}{a_2}-  \frac{-a_2^2 b_0 + a_0 a_1 b_3}{- a_0 a_2}\ge2,$$
we only need
$$ b_0 \le\min\{ \frac{-a_0^2 + 2 a_1 a_2}{a_2^2}, \frac{-a_0^2 - 2 a_0 a_2}{a_2^2}\},\ b_3 \le\frac{a_0^2 a_1 + a_1 a_2^2 b_0}{-a_0^2 a_2}.$$
Here $b_0\le - a_0^2/a_2^2<0$ implies $\min\{\deg(g)\,|\,fg\in \Phi_0\}\ge4$, so $\widehat{\deg}(f)=4=\pi/\theta=\lfloor\pi/\theta\rfloor$.
\item
If $m_0=5$ and $\Delta_{5}=0$, $f(x)=a_2x^2+a_1x+a_0$, where $a_1=-\sqrt{3a_0a_2}$. Rewriting $a_2f(x)=a_2^2x^2+a_2a_1x+3a_1^2$, When we solve \bref{mx-1} for $a_2f(x)$ for $m=6$, we get
\begin{eqnarray*}
b_5<0,\
b_4\le\frac{-a_1^2+3a_1a_2b_5}{3a_2^2},\
\frac{a_1b_4}{a_2}\le b_3\le\frac{3a_1a_2b_4-a_1^2b_5}{3a_2^2}\end{eqnarray*}\begin{eqnarray*}
\frac{2a_1b_3}{3a_2}\le b_2\le\frac{3a_1a_2b_3-a_1^2b_4}{3a_2^2},\
\frac{a_1b_2}{2a_2}\le b_1 \le\frac{3a_1a_2b_2-a_1^2b_3}{3a_2^2},\
\frac{a_1b_1}{3a_2}\le b_0\le\frac{3a_1a_2b_1-a_1^2b_2}{3a_2^2}.
\end{eqnarray*}
Because $b_5<0$ implies $\frac{a_1b_4}{a_2}<\frac{3a_1a_2b_4-a_1^2b_5}{3a_2^2}$, $\frac{a_1b_4}{a_2}< b_3$ implies $\frac{2a_1b_3}{3a_2}<\frac{3a_1a_2b_3-a_1^2b_4}{3a_2^2}$, $\frac{2a_1b_3}{3a_2}< b_2$ implies $\frac{a_1b_2}{2a_2}<\frac{3a_1a_2b_2-a_1^2b_3}{3a_2^2}$, and $\frac{a_1b_2}{2a_2}<b_1$ implies $\frac{a_1b_1}{3a_2}<\frac{3a_1a_2b_1-a_1^2b_2}{3a_2^2}$, there exists a feasible solution $\{b_5,b_4,b_3,b_2,b_1,b_0\}\in\Q_{<0}^6$, which is an inner point of the semi-algebraic set. Using the same notations in \bref{ie-nb}, let $N\in\N_{>1}$ be the common denominator of $\{b_0,\dots,b_5\}$, and we have $f(x)(x^6+N\sum_{j=0}^5b_jx^j)\in\Phi_0$.

Here $b_0<0$ implies $\min\{\deg(g)\,|\,fg\in \Phi_0\}\ge6$, so $\widehat{\deg}(f)=6=\pi/\theta=\lfloor\pi/\theta\rfloor$.
\end{enumerate}
\end{enumerate}
We complete the proof. 
\end{proof}
The following example is used to illustrate the proof.

\begin{exmp}
Let $f=x^2-x+2$, $\Delta_1=-1$, $\Delta_2=1>0$, $m_0=2$, $\widehat{\deg}(f)=2$.
Here $f\notin\N[x]$ implies $\widehat{\deg}(f)>1$, and $(x^2-x+2)(x^2-5x-7)\in\Phi_0$ implies $\widehat{\deg}(f)\le2$.
\end{exmp}

\begin{exmp}
Let $f=x^2-2x+2$. By the effective Polya Theorem \ref{lm-Po}, we have $d_1= \min\{\deg (g)\,|$ $\,g\in\Zx\hbox{ and monic},\  fg\in \Phi_0\}\le10$.
%
%
However, we have $\min\{\deg (g)\,|\,g\in\Zx\hbox{ and monic},\  fg\in \Phi_0\}=4$ by proposition \ref{lm-qu}, where $g=x^4-2x^2-4x-4$ and $fg= x^6-2x^5-8$.
\end{exmp}

\section{Conclusion}

In this paper, we study when a $\sigma$-ideal has a finite $\sigma$-G\"obner basis.
We focused on a special class of $\sigma$-ideals:
normal binomial $\sigma$-ideals which can be be described by
the Gr\"obner basis of a $\Zx$-module.
We give a criterion for a univariate normal binomial $\sigma$-ideal to have a
finite $\sigma$-Gr\"obner basis.
When the characteristic set of the $\sigma$-ideal consists of one $\sigma$-polynomial,
we can give constructive criteria for the $\sigma$-ideal to have a finite $\sigma$-Gr\"obner basis
and an algorithm to compute the finite $\sigma$-Gr\"obner basis under these criteria.
One case is still not solved and we summary it as a conjecture.
Also, it is desirable to extend the criteria given in this paper to multivariate binomial
$\sigma$-ideals.
Example \ref{ex-2} shows that extending Theorem \ref{th-sg1} to the multivariate case is
quite nontrivial.
For $\sigma$-Gr\"obner basis of general $\sigma$-ideals, the work on
monomial $\sigma$-ideals may be helpful \cite{wangjie-1}.

\end{document}